\documentclass[11pt]{article} 
\usepackage[utf8]{inputenc}
\usepackage[a4paper, portrait, margin=1in]{geometry}
\usepackage{color,amsmath,amsfonts,fullpage,amsthm,mathrsfs,fancyhdr,verbatim,framed,hyperref,xspace,mdframed,tabularx,booktabs,enumitem,comment,float}
\hypersetup{
    colorlinks=true,
    linkcolor=blue,
    citecolor=teal,
    hypertexnames=false
}
\usepackage[dvipsnames]{xcolor}
\usepackage{graphicx}
\usepackage{mathpazo}
\usepackage{microtype}
\usepackage{mathtools}
\usepackage{subcaption}
  
\usepackage{caption}
\usepackage{fec} 
\usepackage{soul}
\usepackage{tikz}
\usetikzlibrary{shapes.geometric}
\usepackage[nameinlink, capitalize, noabbrev]{cleveref} 
\usepackage{bm}
\usepackage{bbm}
\newtheorem{theorem}{Theorem}[section]
\newtheorem{lemma}[theorem]{Lemma}
\newtheorem{corollary}[theorem]{Corollary}

\newtheorem{proposition}[theorem]{Proposition}

\usepackage{pifont}
\newtheorem*{problem}{Problem}
\newtheorem{remark}[theorem]{Remark}

\newcommand{\musf}{\boldsymbol{\mu}}

\newcommand{\idx}{\boldsymbol{i}}

\def\ket#1{\mathinner{|{#1}\rangle}}

\usepackage{authblk}

\title{Faster Algorithms for Multimarginal Optimal Transport}
\author[1]{Brandon Augustino\thanks{\texttt{brandon.augustino@jpmchase.com}}}
\author[1]{Yue Sun\thanks{\texttt{yue.sun@jpmorgan.com}}}
\author[1]{Atithi Acharya}
\author[1]{Shouvanik Chakrabarti}
\author[1]{\\\vspace{-.2cm}Junhyung Lyle Kim}
\author[1]{Shree Hari Sureshbabu}
\author[2]{Charlie Che} 
\affil[1]{Global Technology Applied Research, JPMorganChase, New York, NY 10001, USA}
\affil[2]{Quantitative Trading \& Research, JPMorganChase, New York, NY 10179, USA}

\date{}
\begin{document}
\maketitle
\begin{abstract}
We study constructive discrete multimarginal optimal transport (MOT) among $m$ distributions on $n$ points, where the cost tensor $C$ has $N=n^m$ entries. For additive accuracy $\varepsilon$, let $\kappa=\max\{1,(\max C-\min C)/\varepsilon\}$. Classically, we give two algorithms that return exactly feasible additive-$\varepsilon$ couplings. A deterministic box--simplex method with rounding runs in $O(m^2N\kappa\log N)=\widetilde O(m^3N\kappa)$ time, while an exact reduction to positive packing followed by rank-one completion runs in randomized time $\widetilde O(m^2N\kappa)$. At fixed $\kappa$, both match the $\Omega(N)$ cost of writing a dense coupling, up to factors in $m$ and logarithms. Quantumly, in the general entry-access model, tensor Sinkhorn followed by sparse recovery returns an exactly feasible additive-$\varepsilon$ coupling as a classical list of $\widetilde O(m^2n\kappa^2)$ atoms, using $\widetilde O(m^4\sqrt{Nn}\,\kappa^3)$ coherent cost queries without materializing the tensor. Finally, for fixed $m$ and constant normalized accuracy, we prove sparse-output lower bounds of $\widetilde\Omega(N)$ randomized classical queries, even with unrestricted output size, and $\widetilde\Omega(\sqrt{Nn})$ quantum queries for outputs with at most $n\operatorname{polylog}(n)$ atoms. Hence, for fixed $m$ and normalized accuracy between $(\log n)^{-O(1)}$ and a sufficiently small constant, explicit sparse MOT construction has query complexity $\widetilde\Theta(N)$ classically and $\widetilde\Theta(\sqrt{Nn})$ quantumly. 
\end{abstract}

\section{Introduction}
Optimal transport (OT) compares two probability measures through the cheapest way to rearrange the mass of one into the other~\cite{Villani2003Topics}. A transport plan is a joint distribution with the prescribed marginals, and an optimal plan minimizes the expected transportation cost. When the cost is the $p$th power of a metric, the optimal value is the $p$th power of the $p$-Wasserstein distance. OT is now a standard tool in computer vision~\cite{rubner2000earth,bonneel2011displacement,solomon2015convolutional}, statistics~\cite{flamary2018wasserstein,panaretos,szekely2004testing,bigot}, and machine learning~\cite{sandler2011nonnegative,mueller2015principal,lloyd2015statistical,jitkrittum2016interpretable,Arjovsky}.

Multimarginal optimal transport (MOT) couples $m$ measures at once. Given discrete marginals $\mu_k$ on $n_k\le n$ points and a cost tensor $C\in\mathbb{R}^{n_1\times\cdots\times n_m}$, it solves
\begin{equation}\label{e:MOT_intro}\tag{MOT}
    \mathsf{MOT}(C,\musf)
    =\min_{X\in \Ucal (\musf)}
    \langle C,X\rangle, \quad \Ucal(\musf)=\{X\ge0:\proj_k(X)=\mu_k~\forall k\in[m]\}.
\end{equation}
Here $\proj_k(X)$ is the $k$th marginal. Every feasible $X$ has unit mass, so adding a constant to every entry of $C$ changes all feasible costs by the same amount. The relevant scale for additive accuracy $\varepsilon$ is therefore the range $\Delta_C=\max C-\min C$. We write $\kappa=\max\{1,\Delta_C/\varepsilon\}$ for the normalized inverse accuracy.

MOT underlies Wasserstein barycenters~\cite{Agueh2011Barycenters,Cuturi2014Fast,Benamou2015Iterative,Carlier2015Numerical,Srivastava2018Scalable}, density-functional models in statistical physics~\cite{Seidl2007Strictly,Buttazzo2012Optimal,Cotar2013Density,Mendl2013Kantorovich}, matching models in economics~\cite{ekeland2005optimal,carlier2010matching}, robust pricing and hedging in quantitative finance~\cite{Dolinsky2014Robust,Galichon2014Stochastic,henry2017model}, and multi-domain generative modeling~\cite{Choi2018Stargan,Cao2019Multi}.

Each added marginal multiplies the size of the problem. If all supports have
size $n$, then a coupling has $N=n^m$ entries, and writing a dense answer
takes $\Omega(N)$ time. For $m=2$, discrete OT is a minimum-cost flow problem
with specialized algorithms~\cite{schrijver2003combinatorial,
orlin1997polynomial,lee2014path}, but for $m\ge3$ MOT is no longer a flow
problem~\cite{lin2022complexity}. General interior point methods (IPMs) do not
need flow structure. Writing the marginal constraints as
$A\mvec(X)=b$, one has $\operatorname{rank}(A)=mn-m+1$
(Lemma~\ref{lem:sparse_basic_mot}), and the IPM of
Lee and Sidford~\cite{lee2019solving} then returns an exactly feasible
$\varepsilon$-optimal coupling in
$\widetilde O(\poly(m)N\sqrt n)$ arithmetic operations, with only
polylogarithmic dependence on $\kappa$ and on the inverse smallest positive
marginal mass (Remark~\ref{rem:ipm_baseline}). To our knowledge, this
rank-sensitive baseline has not appeared in prior MOT complexity comparisons.
A direct application of current-matrix-multiplication-time LP
solvers gives the much larger bound
$\widetilde O((N+mn)^\omega)$~\cite{cohen2021solving,
van2020deterministic}, where $\omega$ is the matrix-multiplication
exponent. Throughout, $\widetilde O(\cdot)$ hides polylogarithmic factors in
$m$, $n$, $\kappa$, and the inverse failure probability.

For large-scale OT and MOT, first-order methods are often the most practical
alternative to IPMs. Sinkhorn-type scaling uses only marginalization and
elementwise rescaling, and its runtime
can be linear in the explicit tensor size $N$, up to factors in $m$ and
logarithms~\cite{cuturi2013sinkhorn,lin2022complexity}; at fixed regularization, it and its greedy variants converge linearly~\cite{carlier2022linear,kostic2022batch}. Existing MOT bounds trade dependence on $N$ against dependence on $\kappa$. Accelerated alternating minimization~\cite{tupitsa2020multimarginal} is
linear\footnote{These bounds are commonly stated for nonnegative costs using
$\max\{1,\|C\|_\infty/\varepsilon\}$; after the shift
$C\mapsto C-(\min C)\mathbf1$, this parameter is exactly $\kappa$.}
in $\kappa$ but costs $\widetilde O(m^3N\sqrt n\,\kappa)$, whereas
multimarginal Sinkhorn~\cite{lin2022complexity} is linear
in $N$ but quadratic in $\kappa$. Accelerated Sinkhorn~\cite{lin2022complexity}
interpolates with runtime $\widetilde O(m^3Nn^{1/3}\kappa^{4/3})$. None of these bounds is linear in both $N$ and $\kappa$.

We give two constructive algorithms with $\widetilde O(\poly(m)N\kappa)$ arithmetic cost. The first route works over the simplex of all nonnegative
unit-mass tensors and penalizes the total marginal residual
$E(X)=\sum_k\|\proj_k(X)-\mu_k\|_1$ by $2\Delta_C E(X)$. This penalty is the
support function of an $mn$-dimensional box, so the problem is a box--simplex
game. We compute an approximate saddle point, which balances transport
cost against marginal violation, and round its tensor component to exact
feasibility. Multimarginal rounding changes the tensor by at most $2E(X)$ in
$\ell_1$~\cite{lin2022complexity}; because the shifted cost lies in
$[0,\Delta_C]$, its cost increase is at most $2\Delta_C E(X)$, exactly the
amount charged by the penalty. With the box--simplex solver of Jambulapati and Tian~\cite{jambulapati2023revisiting},
this gives a deterministic, exactly feasible algorithm in
$O(m^2N\kappa\log N)=O(m^3N\kappa\log n)$ work.

The second route uses positive packing and rank-one completion and runs in
randomized time
\[
O(N+mn+mN\kappa\log(mnN\kappa)\log(e\kappa))
=\widetilde O(m^2N\kappa).
\]
Both routes match the $\Omega(N)$ dense-output cost at fixed $\kappa$, and
packing saves one explicit factor of $m$. Packing also improves on the
rank-sensitive IPM bound for $\kappa\lesssim\sqrt n$, while the IPM is faster at
higher accuracy (Remark~\ref{rem:ipm_baseline}).

We obtain the packing bound by reducing MOT to a positive packing LP, for which
nearly-linear first-order solvers exist~\cite{allen2019nearly}.
Let $U=\max_{\idx}C_{\idx}$ and replace the cost by the nonnegative benefit
tensor $B=U\mathbf1-C$. Once the marginal constraints are relaxed from $A\mvec(X)=b$ to $A\mvec(X)\leq b$,
every feasible packing solution is a partial coupling. Its
marginal deficits all have the same total mass, so their normalized
tensor product fills every deficit at once and completes the coupling
exactly.

The benefit shift protects the objective during this completion. Because
$B\ge0$, adding the nonnegative completion tensor cannot decrease the packing
benefit, and once the completed tensor has unit mass, its benefit is exactly
$U$ minus its transport cost. The packing optimum therefore equals
$U-\mathsf{MOT}(C,\musf)$, and a relative packing approximation gives additive
MOT accuracy. This extends the bimarginal reductions of Blanchet et al.~\cite{blanchet2024towards} and
Quanrud~\cite{quanrud2019approximating} to arbitrary $m$.

\begin{table}[t]
\caption{Classical arithmetic complexity for an $\varepsilon$-optimal coupling.}
\label{tab:compareOT}
\centering
\begin{tabular}{@{}lll@{}}
\toprule
Method & Reference &Arithmetic runtime \\
\midrule
Interior point &\cite{lee2019solving} & $\widetilde O(\poly(m)\,n^{m+1/2})$ \\
Alternating minimization &\cite{tupitsa2020multimarginal} & $\widetilde O(m^3n^{m+1/2}\kappa)$ \\
Sinkhorn &\cite{lin2022complexity} & $\widetilde O(m^3n^m\kappa^2)$ \\
Accelerated Sinkhorn &\cite{lin2022complexity} & $\widetilde O(m^3n^{m+1/3}\kappa^{4/3})$ \\
\textbf{Box--simplex + rounding} &\textbf{Theorem~\ref{t:box_simplex_main}} & $\boldsymbol{\widetilde O(m^3n^m\kappa)}$ \\
\textbf{Packing + completion} &\textbf{Corollary~\ref{c:classical_main}} & $\boldsymbol{\widetilde O(m^2n^m\kappa)}$ \\
\bottomrule
\end{tabular}
\end{table} 

Quantum algorithms for linear programming~\cite{van2018improvements,bouland2023quantum,apers2023quantum,gao2024logarithmic,kim2026fast} and matrix scaling~\cite{vanapeldoorn2021scaling,gribling2022scaling,nieuwboer2024classical} can run sublinearly in the ambient dimension, so a quantum algorithm might construct an MOT coupling with $o(N)$ cost queries, but existing quantum methods do not. Quantum LP solvers retain polynomial overhead in the $mn$ dual variables and do not directly return a feasible coupling. Existing quantum scaling algorithms handle only matrices; applied to OT, they compute scaling potentials that implicitly specify an approximately feasible solution of entropy-regularized OT. Standard rounding of that solution to an exactly feasible coupling for unregularized OT materializes the dense plan and forfeits the sublinear advantage.

We give such an algorithm in two stages. The first extends quantum Sinkhorn from matrices to order-$m$ tensors. Quantum matrix-scaling algorithms estimate row and column sums with $n$ terms each, whereas each tensor update must estimate $n$ fiber sums with $n^{m-1}$ terms each; we estimate them with a quantum log-sum-exp routine~\cite{gribling2022scaling,nieuwboer2024classical}. These estimates are noisy and depend on all previous estimates, so the exact Sinkhorn identity no longer proves convergence. We replace it with a noisy one-step descent bound and a drift argument showing that most iterates have small marginal error.

The resulting potentials still describe only an implicit, approximately feasible Gibbs tensor, and the multimarginal rounding of Lin et al.~\cite{lin2022complexity} would write this tensor densely. We instead trim the tensor below every target marginal, sample from the trimmed tensor by amplitude amplification, trim the empirical sample once more, and greedily couple the remaining deficits. The final list has exact marginals and can be reused without further cost queries. The whole procedure uses $\widetilde O(m^4\sqrt{Nn}\,\kappa^3)$ coherent cost queries and never writes an $N$-entry tensor.

These upper bounds leave the intrinsic cost of constructive MOT open. The classical methods write a dense tensor, so their linear dependence on $N$ might reflect only serialization, yet every MOT instance has an optimal coupling with at most $\sum_k n_k-m+1$ nonzero entries (Lemma~\ref{lem:sparse_basic_mot}). We therefore also allow algorithms that return only a sparse list of locations and weights. The $\Omega(N)$ cost of writing a dense tensor gives no lower bound for such outputs, so we prove lower bounds for the sparse task itself. For every fixed $m$ and constant accuracy, there are Boolean-cost instances with uniform marginals on which any randomized classical algorithm needs $\Omega(N/\log^2 n)$ queries, whatever the size of its output, and any quantum algorithm returning at most $n\operatorname{polylog}(n)$ atoms needs $\Omega(\sqrt{Nn}/\log n)$ queries. Both bounds persist when exact feasibility is relaxed to marginals with small constant $\ell_1$ error. For fixed $m$ and $\kappa\le\operatorname{polylog}(n)$, the upper and lower query exponents therefore match up to polylogarithmic factors. For estimating only the optimal value, a block-search reduction gives the weaker lower bounds $\Omega(N/(1+\varepsilon n))$ classically and $\Omega(\sqrt{N/(1+\varepsilon n)})$ quantumly (Theorem~\ref{thm:query_lower_bound_mot}); at constant accuracy, these are a factor $n$ below the corresponding sparse-output bounds.

\begin{table}[t]
\caption{Complexity by output type and computational model; query rows count only cost-entry queries. The lower bounds hold for fixed $m$ and normalized accuracy below a small constant; the tight scales hold when, in addition, $\kappa\le\operatorname{polylog}(n)$.}
\label{tab:recovery_complexity}
\centering
\small
\setlength{\tabcolsep}{3.5pt}
\begin{tabular}{@{}lllll@{}}
\toprule
Output & Model and resource & Upper bound & Lower bound & Tight scale \\
\midrule
Dense tensor & Classical arithmetic time
    & $\widetilde O(m^2N\kappa)$ & $\Omega(N)$
    & $\boldsymbol{\widetilde\Theta(N)}$ \\
Sparse list & Classical cost-entry queries
    & $N$ & $\Omega(N/\log^2 n)$
    & $\boldsymbol{\widetilde\Theta(N)}$ \\
Sparse list & Quantum cost queries
    & $\widetilde O(m^4\sqrt{Nn}\,\kappa^3)$
    & $\Omega(\sqrt{Nn}/\log n)$
    & $\boldsymbol{\widetilde\Theta(\sqrt{Nn})}$ \\
\bottomrule
\end{tabular}
\end{table}

Our results apply to arbitrary cost tensors in the general entry-access model, with exact marginals and explicit coupling output; the cost need not have graphical, low-rank, or other succinct structure. When a structured cost can be evaluated coherently, the quantum algorithm runs without materializing the tensor, but the structure may also enable specialized classical methods~\cite{altschuler2023structured,altschuler2021hardness}. We state the separation in cost-query complexity and account separately for the coherent interfaces and the compiled gate complexity.

Section~\ref{s:prelim} fixes notation, and Section~\ref{s:mart} gives the two classical algorithms and the interior point baseline. Section~\ref{s:quantum_tensor_sinkhorn} develops quantum multimarginal Sinkhorn under an explicit finite-input and access model, and Section~\ref{s:sparse_recovery} converts its potentials into an exactly feasible sparse coupling. Section~\ref{s:comparison} compares these methods with other quantum approaches, Section~\ref{s:lower_bounds} proves the sparse-output and value-estimation lower bounds, and Section~\ref{s:limitations_main} discusses limitations. Appendix~\ref{s:quantum} gives a quantum projected-subgradient method that estimates the MOT value and returns an exactly feasible dual certificate using $\widetilde O(m^2n^{m/2+1}\kappa^2)$ cost queries; compared with the constructive quantum algorithm, its dependence on $\kappa$ is better and its dependence on $n$ is worse by a factor $\sqrt n$.

\section{Preliminaries}\label{s:prelim}
We write $[n] \teq \{ 1, \dots, n\}$. We represent the $n \times 1$ all-ones vector by $\mathbf{1}_n \in \R{n}$. We let $\R{n}_{\geq 0}$ (resp. $\R{n}_{> 0}$) denote the nonnegative (resp. positive) orthant, and the probability simplex is $\Delta^n \teq \{ x \in \R{n}_{\geq 0} : \sum_{i \in [n]} x_i = 1 \}$. We define $\Delta^n_{> 0} \teq \R{n}_{> 0} \cap \Delta^n$.  

For a vector $x \in \R{n}$ and $p \in [1, \infty]$, we denote by $\| x \|_p$ the $\ell_p$-norm of $x$. We define the $n$-dimensional $\ell_p$-ball of radius $R$ centered at the origin as $\mathcal{B}_p^n (R) \teq \{ x \in \R{n} : \| x \|_p \leq R \}$. 

For vectors $x, y \in \R{n}$, $\inner{x,y}$ denotes their inner product, and $x \odot y$ denotes their Hadamard (elementwise) product. For matrices $A$ and $B$ of arbitrary sizes, $A \otimes B$ denotes their Kronecker product. For a tensor $A = (A_{i_1, \dots, i_m}) \in \R{n_1 \times \cdots \times n_m}$, $\mvec(A)$ is its $\prod_{k \in [m]} n_k \times 1$ vectorization. When all tensor modes have size $n$, we write a tensor coordinate as a multi-index $\idx=(i_1,\ldots,i_m)\in[n]^m$, and use notation such as $C_{\idx}$, $x_{\idx}$, and $A_{j,\idx}$. For two tensors $A$ and $B$ of equal dimension, their inner product is $\inner{A , B} = \mvec(A)^{\top} \mvec(B)$. We also adopt the convention that the $\ell_p$-norm of a tensor $A$ is the $\ell_p$-norm of its vectorization, i.e., $\| A \|_{p} \teq  \| \mvec(A) \|_{p}.$

\section{Near-linear classical algorithms for MOT}\label{s:mart}
After a shared LP formulation, we present the deterministic box--simplex route, which keeps unit mass, penalizes marginal violations, and rounds an approximate saddle point to exact feasibility. We then give the sharper randomized route, which relaxes the marginal equalities to packing inequalities and completes the resulting partial coupling. The packing reduction extends the bimarginal OT reductions of Blanchet et al.~\cite{blanchet2024towards} and Quanrud~\cite{quanrud2019approximating} and saves one explicit factor of $m$ over the box--simplex route.

\subsection{Finite-dimensional linear programming formulation}
We begin by recalling the discrete MOT problem and its dual:
\begin{problem}
\begin{equation}\tag{MOT-P}\label{e:MOT}
    \begin{aligned}
       \textup{\textsf{MOT}}(C, \musf) \teq\min_{X  \in  \Ucal(\musf)}    \langle C , X \rangle 
    \end{aligned}
\end{equation}
where, with $I_k\teq[n_k]$ and $I\teq I_1\times\cdots\times I_m$, $C = [C_{\idx}]_{\idx\in I}$ denotes a real-valued cost tensor with $m$ modes, $\mu_k = [\mu_{k,i}]_{i \in [n_k]}$ is the probability vector representing the $k$-th marginal and 
$$ \Ucal(\musf) \teq \{X \geq 0 : \proj_k(X)=\mu_k~\forall k\in[m]\}.$$
The dual problem is
\begin{equation}\label{e:MOT_dual_tensor}\tag{MOT-D}
    \begin{aligned}
        \textup{\textsf{MOT}}(C,\musf)
        =\max_{y_1\in\mathbb{R}^{n_1},\ldots,y_m\in\mathbb{R}^{n_m}}\quad
        &\sum_{k=1}^m \mu_k^\top y_k \\
        \textup{subject to}\quad
        &\sum_{k=1}^m y_{k,i_k}\le C_{i_1,\ldots,i_m},
        \qquad i_k\in[n_k]~\forall k\in[m].
    \end{aligned}
\end{equation}
\end{problem}
The product coupling $\mu_1\otimes\cdots\otimes\mu_m$ is primal-feasible, every $X\in\Ucal(\musf)$ has total mass one, and the objective is finite. Hence standard finite-dimensional LP strong duality applies to \eqref{e:MOT}--\eqref{e:MOT_dual_tensor}: for a primal optimal $X^\star\in\Ucal(\musf)$ and a dual optimal $y^\star=(y^\star_1,\ldots,y^\star_m)$, one has
\[
    \sum_{k=1}^m\mu_k^\top y^\star_k=\langle C,X^\star\rangle=\mathsf{MOT}(C,\musf).
\]
Support points with zero marginal mass can be removed without changing the feasible couplings.

\begin{lemma}[Pruning zero marginal coordinates]\label{lem:prune_zero_marginals}
For each mode, let $I_k^+\teq\{a\in I_k:\mu_{k,a}>0\}$. Restrict $C$ to
$I_1^+\times\cdots\times I_m^+$ and restrict each marginal to $I_k^+$. Extension
by zero gives an objective-preserving bijection between couplings of the
restricted and original instances. The same statement holds for the packing
relaxation obtained by replacing every marginal equality by an inequality.
The support sets can be found in $O(\sum_kn_k)$ time, and restriction or zero
extension never increases $\sum_kn_k$, $\prod_kn_k$, or the number of incidence nonzeros.
\end{lemma}
\begin{proof}
If $\mu_{k,a}=0$ and $X$ is a nonnegative feasible coupling, then
\[
    0=\proj_k(X)_a=\sum_{\idx:i_k=a}X_{\idx},
\]
so every entry in that fiber is zero. The same conclusion follows from
$\proj_k(X)_a\le0$ for a packing-feasible tensor. Deleting all such forced
zero fibers therefore preserves feasibility and objective values in both
directions; extension by zero is the inverse map.
\end{proof}

One can cast \eqref{e:MOT} and its dual in standard form using a vectorized representation. Write
$N=\prod_{k=1}^m n_k$, $M=\sum_{k=1}^m n_k$, and define $x=\mvec(X)$. Let
$A\in\{0,1\}^{M\times N}$ be the marginal-incidence matrix with rows indexed by
$(k,r)$ and columns indexed by $\idx=(i_1,\ldots,i_m)$, so that
$A_{(k,r),\idx}=\mathbbm{1}[i_k=r]$. Upon setting
\begin{equation}\label{eq:multi-ot-data}
    c = \mvec (C) \in \mathbb{R}^{N}, 
    \quad b = \mu \teq
    \begin{bmatrix}
        \mu_1 \\ \vdots \\ \mu_m
    \end{bmatrix} \in \mathbb{R}^{M}, 
    \quad A = 
    \begin{bmatrix}
        A_1 \\ 
        \vdots \\
        A_m
    \end{bmatrix}  \in \mathbb{R}^{M \times N},
\end{equation}
the vectorized primal is the standard-form LP
\begin{equation}\tag{P}\label{eq:LP}
      \min_{x\geq 0}\{c^{\top} x: Ax=b\},
\end{equation}
and its dual is
\begin{equation}\tag{D}\label{eq:LP-D}
     \max_{y\in\mathbb{R}^M}\{b^\top y: A^\top y\le c\}.
\end{equation} 
The $M \times N$ constraint matrix $A$ contains $\nnz(A)=mN$ nonzero entries: for each column
$\idx$, exactly one row in each marginal block is active. 

The two classical routes below use this incidence representation differently: box--simplex penalizes the residual $Ax-b$, whereas packing relaxes $Ax=b$ to $Ax\leq b$. Appendix~\ref{s:quantum} uses the same notation for a quantum algorithm that works with \eqref{eq:LP-D}.

\begin{lemma}[Sparse optimal coupling]\label{lem:sparse_basic_mot}
    The MOT linear program has an optimal coupling supported on at most $\operatorname{rank}(A)=\sum_{k=1}^m n_k-m+1$ tensor entries. In particular, when every mode has size $n$, the support has size at most $mn-m+1$.
\end{lemma}
\begin{proof}
    The transportation polytope is nonempty and compact, so a linear objective attains its minimum at a basic feasible solution. The positive columns of a basic feasible solution are linearly independent and hence number at most $\operatorname{rank}(A)$.

    It remains to compute this rank. The sum of the rows in each marginal block is the all-ones row, which gives $m-1$ independent dependencies and therefore
    \begin{equation*}
        \operatorname{rank}(A)\le\sum_{k=1}^m n_k-m+1.
    \end{equation*}
    For the reverse inequality, retain every row in the first block and omit one row from each later block. Suppose a linear combination of the retained rows vanishes. Evaluating it on a tensor index whose later coordinates are all the omitted coordinates first forces every coefficient in the first block to vanish. Varying one later coordinate at a time then forces every remaining coefficient to vanish. The retained rows are linearly independent, proving equality and the support bound.
\end{proof}

\subsection{A deterministic box--simplex route}\label{s:box_simplex}
The box--simplex route preserves unit mass but penalizes marginal error. Let $A$ and $b$ be as in \eqref{eq:multi-ot-data}, define $\Delta_C\teq\max C-\min C$, and shift the cost to
$C^+=C-(\min C)\mathbf1\in[0,\Delta_C]^{I}$. Consider
\begin{equation}\label{e:box_simplex_mot}
    \min_{X\in\Delta^N}\max_{z\in[-1,1]^{M}}
    \mathcal G(X,z),\qquad
    \mathcal G(X,z)
    \teq \langle C^+,X\rangle
    +2\Delta_C z^\top(A\mvec(X)-b).
\end{equation}
For fixed $X$, maximizing over the box gives
$\langle C^+,X\rangle+2\Delta_C E(X)$, where
$E(X)=\sum_k\|\proj_k(X)-\mu_k\|_1$. The multimarginal rounding algorithm of
Lin et al.~\cite[Theorem~4.4]{lin2022complexity} returns $Y\in\Ucal(\musf)$ with
$\|Y-X\|_1\le2E(X)$. Since
$0\le C^+\le\Delta_C$, rounding increases cost by at most
$2\Delta_C E(X)$, which the penalty has already paid for. Thus rounding an approximate saddle point gives an exactly feasible, near-optimal coupling.

\begin{theorem}[Deterministic box--simplex MOT]\label{t:box_simplex_main}
For every $\varepsilon>0$, there is a deterministic algorithm that returns
$Y\in\Ucal(\musf)$ satisfying
$\langle C,Y\rangle\le\mathsf{MOT}(C,\musf)+\varepsilon$ in
$O(m^2N\kappa\log N+mN)$ arithmetic operations; when $n_1=\cdots=n_m=n$, this is $O(m^3n^m\kappa\log n)=\widetilde O(m^3N\kappa)$. The saddle-point solver uses
$O(m\kappa\log N)$ sequential matrix--vector rounds, each with $O(mN)$
work.
\end{theorem}
\begin{proof}
If $\Delta_C\le\varepsilon$, the product coupling $\mu_1\otimes\cdots\otimes\mu_m$ suffices, because every feasible coupling has cost at most $\max C\le\mathsf{MOT}(C,\musf)+\Delta_C$. Otherwise, after negating the objective and exchanging the players, \eqref{e:box_simplex_mot} is the box--simplex form of Jambulapati and Tian~\cite[Theorem~1]{jambulapati2023revisiting}, with box
variable $z\in[-1,1]^{M}$, simplex variable $\mvec(X)\in\Delta^N$, and game
matrix $-2\Delta_C A\in\mathbb R^{M\times N}$.

Let $(\bar X,\bar z)$ have saddle gap at most $\varepsilon$, and let $X^\star$
be an optimal MOT coupling. Since $A\mvec(X^\star)=b$,
\[
\begin{aligned}
    \max_{z\in[-1,1]^{M}}\mathcal G(\bar X,z)
    &\le \min_{X\in\Delta^N}\mathcal G(X,\bar z)+\varepsilon\\
    &\le \mathcal G(X^\star,\bar z)+\varepsilon
    =\mathsf{MOT}(C^+,\musf)+\varepsilon.
\end{aligned}
\]
Maximizing coordinatewise over $z$ shows that
\begin{equation}\label{e:box_simplex_penalty}
    \langle C^+,\bar X\rangle
    +2\Delta_C\sum_{k=1}^m
        \|\proj_k(\bar X)-\mu_k\|_1
    \le \mathsf{MOT}(C^+,\musf)+\varepsilon.
\end{equation}
The rounding algorithm of Lin et al.~\cite[Theorem~4.4]{lin2022complexity} returns
$Y\in\Ucal(\musf)$ satisfying
\[
    \|Y-\bar X\|_1
    \le2\sum_{k=1}^m\|\proj_k(\bar X)-\mu_k\|_1.
\]
Because $0\le C^+\le\Delta_C$, this bound and
\eqref{e:box_simplex_penalty} give
\[
    \langle C^+,Y\rangle
    \le\langle C^+,\bar X\rangle
       +\Delta_C\|Y-\bar X\|_1
    \le\mathsf{MOT}(C^+,\musf)+\varepsilon.
\]
Every feasible coupling has unit mass, so shifting back proves the objective guarantee.

It remains to expose all dimension dependence. Every column of $A$ has exactly
$m$ ones, so
\[
    \nnz(-2\Delta_C A)=mN,
    \qquad
    \|-2\Delta_C A\|_{1\to1}=2m\Delta_C.
\]
The cited box--simplex theorem therefore takes
\[
    O\!\left(
      \nnz(A)\,\|-2\Delta_C A\|_{1\to1}
      \frac{\log N}{\varepsilon}
    \right)
    =O(m^2N\kappa\log N)
\]
work and $O(m\kappa\log N)$ matrix--vector rounds. Each product by $A$ or
$A^\top$ costs $O(mN)$ operations, and multimarginal rounding costs another
$O(mN)$. Finally, $\log N\le m\log n$, so no polynomial dependence on $n$ is
hidden outside $N$.
\end{proof}

\subsection{The packing-LP black box}
The packing route uses the same marginal-incidence data with inequalities and a shifted nonnegative objective; the solver is not applied to the standard-form equality LP. Its multimarginal extension relies on the observation that if a partial coupling $X$ has deficit vectors $d_k\teq\mu_k-\proj_k(X)\in\R{n_k}_{\geq0}$, then all deficits have the same total mass $\tau$, and
\[
    D=\tau^{1-m}d_1\otimes\cdots\otimes d_m
\]
has exactly those deficits as its marginals when $\tau>0$. Thus $X+D$ is an exact coupling. This reduces dense MOT to a packing LP whose constraint matrix has exactly $mN$ nonzero entries. We first state the generic packing solver used by this route.

Allen-Zhu and Orecchia~\cite[Corollary~3.5]{allen2019nearly} state their solver for the standard
form $\max\{\mathbf 1^\top x':A'x'\le\mathbf 1,\ x'\ge0\}$ with $A'\ge0$
normalized so that $\min_j\|A'_{:j}\|_\infty=1$. For accuracy
$\varepsilon'\in(0,1/30]$, it returns, with probability at least $2/3$, a
feasible $(1-O(\varepsilon'))$-approximate solution in expected time
$O(\nnz(A')\varepsilon'^{-1}\log(MN/\varepsilon')\log(1/\varepsilon'))$.
We use it through the following weighted, worst-case form;
cf.~\cite[Theorem~3]{blanchet2024towards}. 
\begin{theorem}[Packing LP solver]\label{t:packing_solver}
    Consider
    \begin{equation}\label{e:packingLP}
        P^{\star} \teq \max \{w^{\top} x : Ax \leq b, x \geq 0 \},
    \end{equation}
    where $A \in \R{M \times N}_{\geq 0}$ has no zero column, $b \in \R{M}_{\geq 0}$, and $w \in \R{N}_{\geq 0}$. For every $\delta_{\rm P} \in (0,1)$ and $\rho\in(0,1/2)$, there is a randomized algorithm which, with probability at least $1-\rho$, returns $x \geq 0$ satisfying $Ax\leq b$ and $w^{\top}x\geq(1-\delta_{\rm P})P^{\star}$ in worst-case arithmetic time
    \[
        O\!\left(\left[M+N+\frac{\nnz(A)}{\delta_{\rm P}}
        \log\!\left(\frac{MN}{\delta_{\rm P}}\right)
        \log\!\left(\frac e{\delta_{\rm P}}\right)\right]\log\frac1\rho\right).
    \]
\end{theorem}
\begin{proof}
    Delete every row $r$ with $b_r=0$ together with every column meeting it, and then every remaining column $j$ with $w_j=0$. Columns of the first kind vanish in every feasible solution. Setting columns of the second kind to zero preserves feasibility, because $A\ge0$, and preserves the objective. The reduced LP therefore has optimum $P^\star$, and zero extension maps its feasible solutions to feasible solutions of \eqref{e:packingLP} with the same objective. For \eqref{e:packingLP_MOT}, the first deletion is exactly Lemma~\ref{lem:prune_zero_marginals}. If no column remains, $x=0$ is optimal. Otherwise substitute
    \[
        x'_j\teq w_jx_j,
        \qquad
        A'_{r,j}\teq\frac{A_{r,j}}{b_rw_j}.
    \]
    This bijection of nonnegative vectors maps $\{Ax\le b\}$ onto $\{A'x'\le\mathbf 1\}$ and $w^\top x$ onto $\mathbf 1^\top x'$, so relative guarantees transfer exactly. Every retained column of $A'$ is nonzero, and multiplying $A'$ by the positive scalar $1/\min_j\|A'_{:j}\|_\infty$ imposes the column normalization while rescaling all solutions and the optimum by the same factor. The reduction does not increase the dimensions or the number of nonzeros, and it and its inverse take $O(M+N+\nnz(A))$ time.

    Run the solver with $\varepsilon'=\min\{c\delta_{\rm P},1/30\}$ for a sufficiently small absolute constant $c>0$, and stop it after three times its expected-time bound. By Markov's inequality and a union bound, each such run returns a feasible $(1-\delta_{\rm P})$-approximate solution with probability at least $1/3$. Checking feasibility and the objective value of an output costs $O(M+N+\nnz(A))$. Therefore $\lceil\log(1/\rho)/\log(3/2)\rceil$ independent runs, keeping the best feasible output, fail with probability at most $\rho$ and obey the stated worst-case bound.
\end{proof}

\subsection{From a packing solution to a multimarginal coupling}
Let $L\teq\min_{\idx\in I}C_{\idx}$ and $U\teq\max_{\idx\in I}C_{\idx}$, both computable in $O(N)$ time, so that $\Delta_C=U-L$.
Define the nonnegative tensor
\begin{equation}
    B \teq U \mathbf{1} - C, \qquad 0 \leq B_{\idx} \leq \Delta_C\quad(\idx\in I).
\end{equation} 
Consider the packing LP 
\begin{equation}\label{e:packingLP_MOT}\tag{packing-MOT}
        P^{\star} \teq \max_{X \in \R{I}_{\geq 0}} \{ \inner{B, X} : \proj_k (X) \leq  \mu_k~\forall k \in [m]\}.
\end{equation}

We begin by showing how feasible solutions of \eqref{e:packingLP_MOT} can be rounded to feasible solutions for \eqref{e:MOT} using $O(mN)$ operations. 
\begin{lemma}\label{lem:rank_one_completion}
    Let $X \in \R{I}_{\geq 0}$ be a feasible solution to \eqref{e:packingLP_MOT}, and define the deficit vectors 
    $$ d_k \teq \mu_k - \proj_k (X) \in \R{n_k}_{\geq 0}.$$
    Then, all deficits have the same total mass
    $$ \| d_k \|_1 = 1 - |X| \teq \tau.$$
    Moreover, there is a nonnegative rank-one tensor $D$, satisfying $\proj_k (D) = d_k$ for every $k$, and hence, $Z = X + D$ belongs to $\Ucal (\musf)$. In particular,
    \begin{equation}\label{e:rank_one_completion}
        D = \begin{cases}
            \tau^{1-m} d_1 \otimes \cdots \otimes d_m & \tau > 0, \\
            0, &\tau = 0
        \end{cases}
    \end{equation}
    can be computed using $O(mN)$ operations.
\end{lemma}
\begin{proof}
    For every $k$,
    $$ \| \proj_k (X) \|_1 = \sum_{a \in I_k} \sum_{\idx: i_k = a} X_{\idx}  = \sum_{\idx \in I} X_{\idx} = |X|.$$
    Since $\| \mu_k \|_1 = 1$ and $d_k \geq 0$, 
    $$ \| d_k \|_1 = 1- |X| = \tau.$$
    
    If $\tau = 0$, then every nonnegative $d_k$ has zero mass, so $d_k = 0$ and $D = 0$ suffices. Now assume $\tau > 0$. For $a \in I_k$, the $a$th entry of the $k$th marginal of the tensor in \eqref{e:rank_one_completion} is 
    $$ (\proj_k (D))_a = \tau^{1-m} d_k (a) \prod_{\ell \neq k} \left( \sum_{a^{\prime} \in I_{\ell} } d_{\ell} (a^{\prime}) \right) = \tau^{1-m} d_k (a) \tau^{m-1} = d_k(a).$$
    Therefore, $\proj_k (X + D) = \proj_k (X) + d_k = \mu_k$ for every $k$. The constraint matrix has $O(mN)$ nonzeros, so computing $\mu-Ax$ and forming the dense tensor product both take $O(mN)$ operations.
\end{proof}

A rank-one deficit correction is also the final step in the multimarginal Sinkhorn rounding of Lin et al.~\cite{lin2022complexity}; their input may have marginal excesses and is first scaled to a subcoupling, whereas packing feasibility gives the submarginal inequalities immediately. Feasibility alone does not protect the objective, because adding mass can be expensive under $C$. The shifted benefit $B$ is nonnegative, so completion can only increase the packing objective, and once the completed tensor has unit mass its benefit is exactly $U$ minus its transport cost. Theorems~\ref{t:packing_vs_mot} and~\ref{t:packing_to_MOT} turn this observation into an exact value correspondence and an additive accuracy guarantee with no second rounding loss.

\subsection{Exact reduction of MOT to a packing LP}
\begin{theorem}\label{t:packing_vs_mot}
    For the MOT value $\mathsf{MOT}(C, \musf)$ and the optimal packing value $P^{\star}$, 
    $$ P^{\star} = U - \mathsf{MOT}(C, \musf).$$
\end{theorem}
\begin{proof}
    Suppose that $X^{\star}$ is an optimal MOT coupling. Then it is feasible for \eqref{e:packingLP_MOT}, has mass one, and 
    $$ \inner{B, X^{\star}} = U - \inner{C, X^{\star}} = U - \mathsf{MOT}(C, \musf).$$
    Hence, $ P^{\star} \geq U - \mathsf{MOT}(C, \musf).$

    Conversely, suppose that $X$ is a feasible solution to \eqref{e:packingLP_MOT}. By Lemma~\ref{lem:rank_one_completion}, there is $D \geq 0$ such that $Z = X + D \in \Ucal(\musf)$. Since $B \geq 0$, 
    $$ \inner{B, X} \leq \inner{B, Z}.$$
    Since $Z$ has mass one and is feasible for \eqref{e:MOT}, 
    $$ \inner{B, Z} = U - \inner{C, Z} \leq U - \mathsf{MOT}(C, \musf).$$
    Thus, every feasible solution to \eqref{e:packingLP_MOT} has packing value at most $U - \mathsf{MOT}(C, \musf)$, so $P^{\star} \leq U - \mathsf{MOT}(C, \musf)$.
\end{proof}

Next, we convert a multiplicative guarantee on the packing value into an additive guarantee on the MOT value. 
\begin{theorem}\label{t:packing_to_MOT}
    Let $\delta_{\rm P} \in (0,1)$, and suppose that $X$ is feasible for \eqref{e:packingLP_MOT} and satisfies
    $$ \inner{B, X } \geq (1 - \delta_{\rm P}) P^{\star}.$$
    Let $Z = X + D$ be its completion from Lemma~\ref{lem:rank_one_completion}. Then, $Z \in \Ucal (\musf)$ and
    $$ \inner{C, Z} \leq \mathsf{MOT}(C, \musf) + \delta_{\rm P} \Delta_C.$$
\end{theorem}

\begin{proof}
    Since $D \geq 0$ and $B \geq 0$,
    $$ \inner{B , Z } \geq \inner{ B , X } \geq (1-\delta_{\rm P}) P^{\star}.$$
    Applying Theorem~\ref{t:packing_vs_mot} and using the fact that $|Z| = 1$, we have
    \begin{align*}
        \inner{ C ,Z} - \mathsf{MOT}(C, \musf)
        = \left[ U - \inner{B , Z } \right] - (U - P^{\star}) = P^{\star} - \inner{B , Z }  \leq \delta_{\rm P} P^{\star}.
    \end{align*}
    Since $\mathsf{MOT}(C, \musf) \geq L$, the exact equivalence gives $P^{\star} = U -  \mathsf{MOT}(C, \musf) \leq U - L = \Delta_C$, so $\delta_{\rm P}P^\star\le\delta_{\rm P}\Delta_C$.
\end{proof}

Combining Theorems~\ref{t:packing_solver} and~\ref{t:packing_to_MOT} gives the main classical result. 

\begin{theorem}\label{t:classical_general}
    Let $m \geq 2$ and $\Delta_C \teq \max_{\idx} C_{\idx} - \min_{\idx} C_{\idx}$. For every $\varepsilon > 0$ and $\rho\in(0,1/2)$, one can, with probability at least $1-\rho$, determine a tensor $Z \in \Ucal (\musf)$ satisfying $\inner{C,Z}\leq\mathsf{MOT}(C,\musf)+\varepsilon$. Set $\kappa\teq\max\{1,\Delta_C/\varepsilon\}$. This can be done in randomized arithmetic time
    \[
        O\!\left(\left[N+M+mN\kappa
        \log(MN\kappa)\log(e\kappa)\right]
        \log\frac1\rho\right).
    \]
\end{theorem}
\begin{proof}
    Let $L = \min_{\idx} C_{\idx}$ and $U = \max_{\idx} C_{\idx}$, so that $\Delta_C = U-L$. If $\varepsilon \geq \Delta_C$, return the product coupling $Z^0 = \mu_1 \otimes \cdots \otimes \mu_m$, which can be formed in $O(N)$ time. Since $Z^0 \in \Ucal(\musf)$, $\inner{C,Z^0} \leq U$ and $\mathsf{MOT}(C, \musf) \geq L$, so
    $$ \inner{C,Z^0} - \mathsf{MOT}(C, \musf) \leq U - L = \Delta_C \leq \varepsilon.$$

    Assume $\varepsilon \in (0, \Delta_C)$ and set $\delta_{\rm P} = \varepsilon / \Delta_C \in (0,1)$. The packing LP \eqref{e:packingLP_MOT} has $M$ constraints, $N$ variables, and a coefficient matrix with $mN$ nonzero entries and no zero column. The black box packing algorithm of Theorem~\ref{t:packing_solver} returns, with probability at least $1-\rho$, a feasible solution $X$ to \eqref{e:packingLP_MOT} in time
    $$ O \left(\left[N+M+mN\frac{\Delta_C}{\varepsilon}\log\!\left(MN\frac{\Delta_C}{\varepsilon}\right)\log\!\left(\frac{e\Delta_C}{\varepsilon}\right)\right]\log\frac1\rho\right).$$
    By Theorem~\ref{t:packing_to_MOT}, the completion tensor $Z$ is an exactly feasible solution to \eqref{e:MOT}, and has cost not exceeding $ \mathsf{MOT}(C, \musf) + \varepsilon$.
\end{proof}

\begin{corollary}[Near-linear classical MOT]\label{c:classical_main}
    Suppose that $n_1=\cdots=n_m=n$. For every $\varepsilon>0$ and $\rho\in(0,1/2)$, let $\kappa=\max\{1,\Delta_C/\varepsilon\}$. There is a randomized algorithm that, with probability at least $1-\rho$, returns $Z\in\Ucal(\musf)$ satisfying $\langle C,Z\rangle\le\mathsf{MOT}(C,\musf)+\varepsilon$ using
    \[
        O\!\left(m^2n^m\kappa
        \log(emn\kappa)\log(e\kappa)\log\frac1\rho\right)
        =\widetilde O\!\left(m^2n^m\kappa\log\frac1\rho\right)
    \]
    arithmetic operations.
\end{corollary}
\begin{proof}
Apply Theorem~\ref{t:classical_general} with $N=n^m$ and $M=mn$. Since $N+M=O(mN)$ and $\log(MN\kappa)\le(m+2)\log(emn\kappa)$, the bound follows.
\end{proof}

Algorithm~\ref{alg:classical_main} summarizes the method.

\begin{algorithm}[t]
\caption{MOT via packing and rank-one completion}
\label{alg:classical_main}
\begin{algorithmic}[1]
\Require Marginals $\musf$, cost access to $C$, accuracy $\varepsilon$
\State Compute $L,U$ and $B=U\mathbf1-C$
\If{$\varepsilon\ge U-L$}
    \State \Return $\mu_1\otimes\cdots\otimes\mu_m$
\EndIf
\State $X\gets\textsc{Packing}(B,\musf,\varepsilon/(U-L))$
\State $d_k\gets\mu_k-\proj_k(X)$ for every $k$ and $\tau\gets\|d_1\|_1$
\State $D\gets0$ if $\tau=0$, else $D\gets\tau^{1-m}\bigotimes_{k=1}^m d_k$
\State \Return $X+D$
\end{algorithmic}
\end{algorithm}

\subsection{An interior point baseline}\label{s:ipm_baseline}
\begin{remark}[Interior point baseline]\label{rem:ipm_baseline}
We record the following consequence of the rank identity in Lemma~\ref{lem:sparse_basic_mot}; to our knowledge, this specialization has not previously appeared in MOT complexity comparisons. For $m\ge3$ the MOT linear program is not a minimum-cost flow problem~\cite{lin2022complexity}, but general interior point methods do not need that structure. Lee and Sidford~\cite[Theorem~1]{lee2019solving} consider linear programs $\min\{c^\top x:A^\top x=b,\ l\le x\le u\}$ in which $A$ has full column rank and no zero row. Given a point $x_0$ with $A^\top x_0=b$ and $l<x_0<u$, their algorithm returns, with constant probability, a point $x$ with $A^\top x=b$, $l<x<u$, and $c^\top x\le\mathrm{OPT}+\epsilon$. It uses $\widetilde O(\sqrt{\operatorname{rank}(A)}\log(U_{\rm LS}/\epsilon))$ iterations, where $U_{\rm LS}$ bounds $\|c\|_\infty$, $\|u-l\|_\infty$, and the inverse distance from $x_0$ to the bounds. Each iteration performs $\widetilde O(1)$ solves with $A^\top DA$ for positive diagonal $D$, plus nearly linear work.

For MOT, take $x=\mvec(X)$, keep all marginal equations of the first block and all but one of every later block, so that the retained $mn-m+1$ equations are independent (Lemma~\ref{lem:sparse_basic_mot}), and use the nonbinding bounds $0\le X\le 2$. After pruning zero-mass points (Lemma~\ref{lem:prune_zero_marginals}), the product coupling $\mu_1\otimes\cdots\otimes\mu_m$ is a valid starting point, so $\log U_{\rm LS}=O(m\log(1/\mu_{\min})+\log(1+\|C\|_\infty))$, where $\mu_{\min}$ is the smallest positive marginal mass. The output is an exactly feasible coupling, in the same exact-arithmetic model as Theorem~\ref{t:classical_general}, with cost at most $\mathsf{MOT}(C,\musf)+\varepsilon$. The system matrix has dimension $mn-m+1$. It is formed in $O(m^2N)$ operations and factored in $O((mn)^\omega)=O(\poly(m)N)$ operations for $m\ge3$, and the algorithm's polylogarithmic factors in the number $N$ of variables are polynomial in $m$. Repeating $O(\log(1/\rho))$ times and keeping the cheapest feasible output gives failure probability $\rho$ in
\[
    \widetilde O\!\left(\poly(m)\,n^{m+1/2}\log\frac{1}{\mu_{\min}}\right)
\]
arithmetic operations. For $m=2$, where the system matrices are equivalent to bipartite-graph Laplacians, the same approach gives the $\widetilde O(n^{2.5})$ interior point bound that Blanchet et al.~\cite{blanchet2024towards} list for optimal transport. Compared with Theorem~\ref{t:classical_general}, the packing algorithm is faster when $\kappa\lesssim\sqrt n$, up to $\poly(m)$ and logarithmic factors, and the interior point method is faster at higher accuracy. In particular, the $\widetilde O(\sqrt{mn}\,n^m)$ bound that Lin et al.~\cite{lin2022complexity} state conditionally on a min-cost-flow formulation holds without that structure, up to $\poly(m)$ factors. These factors are large: the work bound of Lee and Sidford~\cite{lee2019solving} carries a factor $\log^{13}N=(m\log n)^{13}$. For moderate $\kappa$, first-order methods such as accelerated alternating minimization~\cite{tupitsa2020multimarginal}, whose dependence on $m$ is $m^3$, therefore remain faster in $m$ and are far simpler to implement; the interior point bound is chiefly of theoretical interest.
\end{remark}

\section{Quantum Sinkhorn for multimarginal OT}\label{s:quantum_tensor_sinkhorn}
Writing all $N=n^m$ entries of a coupling already costs $N$ time, so a sublinear quantum algorithm must return a smaller object. This section develops the first stage of our constructive quantum algorithm, a quantum multimarginal Sinkhorn method for entropy-regularized MOT. Entropic regularization turns the constrained primal problem into a smooth scaling problem whose solution has a compact description: $\sum_kn_k$ scaling potentials specify a Gibbs tensor without listing its $N$ entries. Section~\ref{s:sparse_recovery} converts these potentials into an exactly feasible sparse coupling.

\subsection{Access model and finite-input conventions}\label{s:finite_input_appendix}
Costs, marginals, and accuracy are given as finite binary numbers, and exact
feasibility refers to these input marginals. We assume coherent cost access
\begin{equation}\label{e:cost_oracle_main}
    \mathcal O_C:\ket{\idx}\ket{w}\longmapsto\ket{\idx}\ket{w\oplus C_{\idx}}.
\end{equation}
We also assume coherent lookup of the $O(mn)$ entries of the classical working
vectors (potentials and trimming factors) and marginals, and preparation of $\sum_a\sqrt{\mu_{k,a}}\ket a$. The stated bounds
count calls to these interfaces; Remark~\ref{rem:no_qram} gives their compiled
gate costs, including ordinary-table access.

An $(s,t)$ signed fixed-point number is $z2^{-t}$ with
$z\in\mathbb Z$ and $|z|<2^{s+t}$, stored using one sign bit, $s$ integer
bits, and $t$ fractional bits. A $b$-bit number has total length at most $b$.
For the constructive quantum upper bounds, each $C_{\idx}$ is an exact signed
$b_C$-bit fixed-point number, each $\mu_{k,a}$ is a nonnegative $b_\mu$-bit
dyadic rational satisfying $\sum_a\mu_{k,a}=1$ exactly, and
$\varepsilon>0$ is an exact $b_\varepsilon$-bit dyadic rational. All
optimization and exact-feasibility guarantees refer to these represented data.
If $\|C-C^\circ\|_\infty\le\xi_C$, the same output is
$(\varepsilon+2\xi_C)$-optimal for $C^\circ$ with the represented marginals;
approximating the marginals instead requires a separate stability or rounding
argument.

Assume $n\ge2$, since the one-point case is trivial. For $x>0$, define
$\operatorname{dyad}(x)\teq2^{\lfloor\log_2x\rfloor}$, so
$x/2<\operatorname{dyad}(x)\le x$. Let $q_n\teq\lceil\log_2 n\rceil$ and
fix the exactly represented full-support distribution
\[
    u^{\rm dyad}_{n,a}\teq
    \begin{cases}
        2^{-q_n},&a<n,\\
        1-(n-1)2^{-q_n},&a=n,
    \end{cases}
    \qquad
    \min_a u^{\rm dyad}_{n,a}\ge\frac1{2n}.
\]
In \eqref{e:cost_oracle_main}, $\oplus$ is bitwise XOR with the canonical
fixed-width encoding. Elementary functions, controlled rotations, and state
preparation are approximated to the tolerances assigned below. If
$C_{\idx}=\mathbf c(z_{1,i_1},\ldots,z_{m,i_m})$ for represented support points $z_{k,a}$ and a cost function $\mathbf c$, the oracle can load the
support points, evaluate $\mathbf c$ coherently to $C_{\idx}$,
and uncompute them. Thus the tensor itself need not be stored in qRAM.

The generic scaling theorem in this section accepts coherent entry access to a positive tensor $K$ with known bounds $K_{\min}$ and $K_{\max}$. For MOT we use log-domain access instead: every regularization parameter below is a power of two, so the Gibbs exponent is exact and the kernel $K_{\idx}=\exp(-C_{\idx}/\eta)$ is never evaluated.

\begin{lemma}[Exact log-domain Gibbs exponents]\label{lem:cost_to_gibbs}
Let each $C_{\idx}$ be an exact signed fixed-point number with at most $b_C$
bits, let $\eta=2^{-e}$ for an integer $e$, and write $b_\eta\teq|e|+1$.
Let $\beta$ be fixed-point potentials on a common grid whose entries use at
most $\ell_\beta$ bits. Then $O(1)$ calls to \eqref{e:cost_oracle_main},
$O(m)$ potential lookups, and $O(m(b_C+b_\eta+\ell_\beta))$ elementary gates
compute
\[
    w_{\idx}(\beta)\teq\sum_{k=1}^m\beta_{k,i_k}-\frac{C_{\idx}}{\eta}
\]
exactly, as a fixed-point number with $O(b_C+b_\eta+\ell_\beta+\log m)$
bits. In particular, the Gibbs weight
$K_{\idx}e^{\sum_k\beta_{k,i_k}}=e^{w_{\idx}(\beta)}$ of the exact kernel
$K_{\idx}=e^{-C_{\idx}/\eta}$ has an exactly represented exponent.
\end{lemma}
\begin{proof}
Dividing by $\eta=2^{-e}$ multiplies $C_{\idx}$ by $2^e$, which moves the
binary point by $e$ positions. Hence $C_{\idx}/\eta$ is exact and fits in a
register with $|e|$ additional integer or fractional positions; forming it is
a relabeling of wires together with zero padding. Adding the $m$ potentials
and $-C_{\idx}/\eta$ on a common grid, in a register with
$\lceil\log_2(m+1)\rceil$ additional integer bits, costs
$O(m(b_C+b_\eta+\ell_\beta))$ gates with ripple-carry adders.
Compute--copy--uncompute makes the circuit reversible with $O(1)$ cost-oracle
calls and $O(m)$ lookups.
\end{proof}

Because the exponent is exact, Lemma~\ref{lem:logsumexp} below can be applied
with $a\equiv1$ and the exponents $w_{\idx}(\beta)$. Its register widths grow
only logarithmically with the range of these exponents, which is dominated by
$\|C\|_\infty/\eta$ and the potential range, whereas an entrywise fixed-point
representation of $K_{\idx}$ would need $\Theta(\|C\|_\infty/\eta)$
fractional bits. Errors are therefore measured against the exact Gibbs
tensor, and all approximation is confined to the log-sum-exp estimates of
Lemma~\ref{lem:logsumexp}. Using $r=\mu_{k,a}$ as in
Lemma~\ref{lem:quantum_tensor_marginals} incorporates the target log-mass
without evaluating it separately, and subtraction and recentering are exact
on the common fixed-point grid.

\subsection{Entropic MOT and its associated dual}
Define the discrete Shannon entropy $H(X)\teq -\sum_{\idx}X_{\idx}\log X_{\idx}$, with the convention $0\log0=0$. This differs from the shifted entropy convention used in \cite{lin2022complexity} by an additive constant on probability tensors, and therefore leaves the scaling potentials and all marginal equations unchanged. Fix a regularization parameter $\eta>0$ and consider the \textit{entropic regularized} MOT problem:
\begin{equation}\label{e:entropic_mot_tensor}\tag{$\mathrm{MOT}_\eta$}
    \min_{X\in\Ucal(\musf)}
    \left\{\langle C,X\rangle-\eta H(X)\right\}.
\end{equation}
Let $\lambda=(\lambda_1,\ldots,\lambda_m)$ be the Lagrange multipliers for the marginal constraints. Keeping the redundant mass constraint $\|X\|_1=1$ while dualizing the marginal equations, set
\[
    \beta_k\teq \eta^{-1}\lambda_k,
    \qquad k\in[m].
\]
With this change of variables, and after changing sign and dividing the dual objective by $\eta$, the entropic dual is the unconstrained smooth convex problem
\begin{equation}\tag{$\mathrm{MOT}_\eta$-D}\label{e:entropic_mot_dual_scaling}
    \min_{\beta_1\in\mathbb{R}^{n_1},\ldots,\beta_m\in\mathbb{R}^{n_m}}
    \Phi_\eta(\beta)
    \teq
    \log\!\left(\sum_{\idx}
    \exp\!\left(\sum_{k=1}^m\beta_{k,i_k}-\frac{C_{\idx}}{\eta}\right)
    \right)
    -\sum_{k=1}^m\mu_k^\top\beta_k.
\end{equation}
Up to the irrelevant additive constant coming from the entropy convention, this is the dual formulation used in \cite[Section~2.2]{lin2022complexity}. 

Assuming the listed support points have positive marginal masses, the entropy term in \eqref{e:entropic_mot_tensor} makes the primal objective strictly convex and ensures that the unique optimizer $X^{\star}$ is strictly positive on the product support. The Karush-Kuhn-Tucker (KKT) optimality conditions imply that this solution admits the Gibbs representation
\begin{equation}\label{e:entropic_mot_scaling_form}
    X^{\star}_{\idx}
    =\widehat K(\beta)_{\idx}
    \teq
    \frac{K_{\idx}\exp\!\left(\sum_{k=1}^m \beta_{k,i_k}\right)}
    {\sum_{\idx'}K_{\idx'}\exp\!\left(\sum_{k=1}^m \beta_{k,i'_k}\right)},
    \qquad
    K_{\idx}\teq \exp(-C_{\idx}/\eta),
\end{equation}
for log-scaling potentials $\beta_k\in\mathbb{R}^{n_k}$ chosen so that $\proj_k(X)=\mu_k$ for all $k \in [m]$. Equivalently, if
\[
    K(\beta)_{\idx}\teq K_{\idx}\exp\!\left(\sum_{k=1}^m \beta_{k,i_k}\right),
\]
then solving \eqref{e:entropic_mot_tensor} is exactly the problem of scaling the positive Gibbs tensor $K$ so that $\widehat K(\beta)=K(\beta)/\|K(\beta)\|_1 \in \Ucal (\musf)$. After a gauge shift of the potentials, one may also impose $\|K(\beta)\|_1=1$ and view the optimizer as the unnormalized scaled tensor $K(\beta)$. The potentials $\lambda$ and $\beta$ are unrelated to the LP dual variables $y$ of Section~\ref{s:mart}. 

Since $\Phi_\eta(\beta)=\log\|K(\beta)\|_1-\sum_k\mu_k^\top\beta_k$, the block gradients of the entropic dual are the marginal errors:
\[
    \nabla_{\beta_k}\Phi_\eta(\beta)=\proj_k(\widehat K(\beta))-\mu_k,
    \qquad
    \widehat K(\beta)\teq \frac{K(\beta)}{\|K(\beta)\|_1}.
\]
The multimarginal Sinkhorn iterate of \cite[Algorithm 1]{lin2022complexity} is a block coordinate step on \eqref{e:entropic_mot_dual_scaling}: it chooses a block by a greedy KL-residual rule and updates that block so that the corresponding marginal matches its target. We use the same one-block update in normalized, gauge-equivalent notation but choose the block uniformly at random.

\subsection{Quantum multimarginal Sinkhorn}
We now state the randomized coordinate-descent iteration that the quantum algorithm implements. For the complexity statements in this subsection, assume $n_1=\cdots=n_m=n$. Let $K\in\mathbb{R}_{>0}^{[n]^m}$ be a positive tensor and let $\mu_k\in\Delta^n_{>0}$ be the target marginals. In this subsection we write
\[
    Z(\beta)
    \teq
    \sum_{\idx}K_{\idx}\exp\!\left(\sum_{\ell=1}^m\beta_{\ell,i_\ell}\right),
    \qquad
    \Phi_K(\beta)
    \teq
    \log Z(\beta)-\sum_{\ell=1}^m\mu_\ell^\top\beta_\ell.
\]
For the Gibbs tensor $K_{\idx}=\exp(-C_{\idx}/\eta)$, this is exactly the dual objective $\Phi_\eta$ in \eqref{e:entropic_mot_dual_scaling}. The normalized scaled tensor is
\[
    \widehat K(\beta)_{\idx}
    \teq
    \frac{K_{\idx}\exp(\sum_\ell\beta_{\ell,i_\ell})}{Z(\beta)}.
\]
The current $k$-th marginal and total marginal error are
\[
    q_k(\beta)\teq\proj_k(\widehat K(\beta)),
    \qquad
    E(\beta)\teq\sum_{k=1}^m\|q_k(\beta)-\mu_k\|_1.
\]

Starting from $\beta^{(0)} = 0$, at iteration $t$ we sample a mode
\[
    k_t\sim\mathrm{Unif}([m]).
\]
For a selected mode $k$, the exact one-block update is
\begin{equation}\label{e:exact_random_sinkhorn_update}
    \beta_{k,a}^{+}
    =\beta_{k,a}+\log\mu_{k,a}-\log q_{k,a}(\beta),
    \qquad
    \beta_{\ell}^{+}=\beta_\ell\quad(\ell\neq k).
\end{equation}
At iteration $t$, we apply \eqref{e:exact_random_sinkhorn_update} with $k=k_t$ and $\beta=\beta^{(t)}$ and then recenter the potentials. The update sets the selected marginal to $\mu_{k_t}$ and leaves the other marginals to later updates.

Each iteration therefore needs the vector $\log q_{k_t}(\beta^{(t)})$, which the quantum algorithm estimates without summing over every fiber entry. For a fixed mode $k$ and coordinate $a$, define the unnormalized fiber mass
\[
    S_{k,a}(\beta)
    \teq
    \sum_{\idx:i_k=a}K_{\idx}\exp\!\left(\sum_{\ell=1}^m\beta_{\ell,i_\ell}\right),
    \qquad
    q_{k,a}(\beta)=\frac{S_{k,a}(\beta)}{Z(\beta)}.
\]
Therefore
\[
    \log q_{k,a}(\beta)=\log S_{k,a}(\beta)-\log Z(\beta).
\]
The first term is a log-sum-exp over a fiber of size $n^{m-1}$, and the second is a log-sum-exp over the whole tensor. We estimate $\log S_{k_t,a}(\beta^{(t)})$ for every $a$ and $\log Z(\beta^{(t)})$ with a quantum subroutine, perform the Sinkhorn update with the resulting noisy log-marginals, and prove that this noisy randomized coordinate descent reaches $E(\beta)\le\delta$. Lemma~\ref{lem:logsumexp} states the quantum \textsc{LogSumExp} primitive for estimating quantities of the form 
$$\log\left(\frac{1}{r}\sum_{j=1}^d a_j\exp(w_j)\right).$$
Its query algorithm and query bound specialize Nieuwboer~\cite[Theorem~13.3.2]{nieuwboer2024classical}; the statement and proof below make the normalization by $r$ and the reversible fixed-point gate accounting explicit.
Lemma~\ref{lem:quantum_tensor_marginals} turns this subroutine into a one-mode marginal oracle, Lemma~\ref{lem:noisy_sinkhorn_descent} analyzes one exact or noisy update, and Theorem~\ref{t:quantum_tensor_sinkhorn} combines these ingredients into the final complexity bound.

\begin{lemma}[Approximate log-sum-exp; explicit specialization of~{\cite{nieuwboer2024classical}}]\label{lem:logsumexp}
Let $d\ge1$. Let $a\in[\gamma,1]^d\cap\mathbb Q^d$ and
$r\in(0,1]\cap\mathbb Q$ have exact binary encodings of at most
$\ell_a$ and $\ell_r$ bits, respectively. Let $w\in[-B,B]^d$ be stored
exactly in signed fixed-point format with $b_I$ integer bits (including the
sign bit) and $b_F$ fractional bits. Assume reversible value-oracle access
\[
    \mathcal O_a:\ket j\ket z\longmapsto\ket j\ket{z\oplus a_j},
    \qquad
    \mathcal O_w:\ket j\ket z\longmapsto\ket j\ket{z\oplus w_j}.
\]
Let $\delta,\zeta\in(0,1)\cap\mathbb Q$ have exact encodings of lengths
$\ell_\delta$ and $\ell_\zeta$, and define
\[
    L\teq\log\left(\frac1r\sum_{j=1}^d a_j e^{w_j}\right).
\]
Suppose the output register has at least
\begin{equation}\label{e:lse_output_width}
    b_I\ge
    2+\left\lceil\log_2\!\left(1+B+\log\frac{d}{\gamma r}\right)\right\rceil,
    \qquad
    b_F\ge\left\lceil\log_2\frac8\delta\right\rceil.
\end{equation}
Then a quantum algorithm outputs a fixed-point number $\Lambda$ satisfying
$|\Lambda-L|\le\delta$ with probability at least $1-\zeta$, using
\[
    O\!\left(\frac{\sqrt d}{\delta}\log\frac1\zeta\right)
\]
total calls to $\mathcal O_a$ and $\mathcal O_w$. If $\mathsf{Mult}(q)$ is the gate
cost of multiplying two $q$-bit integers and
\begin{equation}\label{e:logsumexp_word_length}
    b_{\rm LSE}
    \teq
    O\!\left(
        \ell_a+\ell_r+\ell_\delta+\ell_\zeta+b_I+b_F+
        \log\frac{d}{\delta\gamma r}
    \right),
\end{equation}
then the non-oracle elementary-gate complexity is
\[
    \widetilde O\!\left(
        \frac{\sqrt d}{\delta}\log\frac1\zeta\,
        \mathsf{Mult}(b_{\rm LSE})
    \right).
\]
In particular, quadratic fixed-point multiplication gives the same bound
with $\mathsf{Mult}(b_{\rm LSE})=O(b_{\rm LSE}^2)$.
\end{lemma}
\begin{proof}
We apply Nieuwboer~\cite[Theorem~13.3.2]{nieuwboer2024classical} with the following
dictionary between its parameters and ours:
\[
    \gamma_0\teq\min\left\{\gamma,\frac r2\right\},
    \qquad
    (n,a,r,y,\text{precision},\text{failure},\mu)
    =(d,a,r,w,\delta,\zeta,\gamma_0).
\]
Its hypotheses require exact rational encodings of $a$ and $r$, fixed-point
access to $w$, and a lower bound $\mu$ on every nonzero entry of $a$; these
are precisely the assumptions above. The smaller source lower bound costs at
most an additional $O(\log(1/(\gamma r)))$ bits, already included in
\eqref{e:logsumexp_word_length}. Moreover,
\[
    \frac{d\gamma e^{-B}}r
    \le e^L\le
    \frac{de^B}r,
    \qquad
    |L|\le B+\log\frac{d}{\gamma r}.
\]
Thus \eqref{e:lse_output_width} supplies the integer and fractional output
widths required by the cited theorem. Its query bound is exactly
$O(\sqrt d\,\delta^{-1}\log(1/\zeta))$.

For completeness, we expose the normalization used by that theorem. Its
Algorithm~13.3 first rounds each rational $a_j$ upward using
\[
    b_a=\left\lceil\log_2\frac{4}{\delta\gamma_0}\right\rceil
\]
fractional bits. Lemma~13.3.8 there shows that this changes the target
logarithm by at most $\delta/4$. Approximate quantum maximum finding then
returns $j^\star$ such that
\[
    \frac{a_j}{a_{j^\star}}e^{w_j-w_{j^\star}}\le\frac32
    \qquad(j\in[d]).
\]
The algorithm coherently approximates and clips one half of these relative
weights into $[0,1]$, applies quantum approximate summation at relative
accuracy $\Theta(\delta)$, and reconstructs
\[
    L=w_{j^\star}+\log a_{j^\star}
      +\log\!\left(\sum_{j=1}^d
        \frac{a_j}{a_{j^\star}}e^{w_j-w_{j^\star}}\right)-\log r.
\]
The error allocation in the proof of Theorem~13.3.2 assigns the remaining
budget to relative-weight approximation, approximate summation, evaluation
of the three logarithms, and final fixed-point truncation, yielding total
additive error at most $\delta$.

It remains only to make explicit the bit cost hidden by the cited theorem.
The rounded weights use $O(\log(1/(\delta\gamma_0)))$ fractional bits, the
approximate sum uses $O(\log(d/\delta))$ bits, and $r$, $w$, and the output
use their declared encodings. Hence every ratio, range-reduced exponential,
product, logarithm, and comparison in Algorithm~13.3 fits in
$O(b_{\rm LSE})$ bits. Division, exponentiation, and logarithm evaluation
cost $\widetilde O(\mathsf{Mult}(b_{\rm LSE}))$ gates~\cite[Chapter~4]{brent2010modern},
and compute--copy--uncompute makes these circuits reversible with the same
asymptotic cost. Multiplying by the cited number of coherent evaluations
gives the displayed gate bound.
\end{proof}

\begin{lemma}[Quantum marginal estimation]\label{lem:quantum_tensor_marginals}
    Suppose $n_1=\cdots=n_m=n$, every entry of a positive rational tensor
    $K$, as well as the known rational bounds $K_{\min}$ and $K_{\max}$,
    has an exact encoding of at most $\ell_K$ bits, and
    $0<K_{\min}\le K_{\idx}\le K_{\max}$ hold for every
    $\idx\in[n]^m$. Assume exact coherent entry access to $K$ and coherent
    lookup of fixed-point potentials $\beta$ whose entries use at most
    $\ell_\beta$ bits. Let $\|\beta\|_\infty\le B$ and fix $k\in[m]$.
    For represented rational $\alpha,p\in(0,1)$ with encoding lengths
    $\ell_\alpha$ and $\ell_p$, there is a quantum algorithm which outputs
    numbers $\widetilde\ell_{k,a}$ satisfying
    \[
        \left|\widetilde\ell_{k,a}-\log q_{k,a}(\beta)\right|\le\alpha,
        \qquad a\in[n],
    \]
    with probability at least $1-p$, using $\widetilde O(n^{(m+1)/2}\alpha^{-1})$ coherent kernel queries and $\widetilde O(mn^{(m+1)/2}\alpha^{-1})$ potential lookups and unit-cost arithmetic
    operations. If
    \[
        b_{\rm marg}
        \teq O\!\left(
            \ell_K+\ell_\beta+\ell_\alpha+\ell_p
            +\log\frac{K_{\max}}{K_{\min}}
            +m\log n+\log\frac{m}{\alpha p}+\log(1+B)
        \right),
    \]
    then its non-oracle elementary-gate complexity is
    $\widetilde O(mn^{(m+1)/2}\alpha^{-1}\mathsf{Mult}(b_{\rm marg}))$.
    Given in addition an exact rational target mass $\mu_{k,a}>0$, the same
    bounds, with its encoding length and $\log(1/\mu_{k,a})$ included in
    $b_{\rm marg}$, produce an $\alpha$-accurate estimate of the complete
    correction $\log q_{k,a}(\beta)-\log\mu_{k,a}$. Finally, suppose instead
    that $K_{\idx}=e^{-C_{\idx}/\eta}$ with the costs and $\eta=2^{-e}$ of
    Lemma~\ref{lem:cost_to_gibbs}, that $0\le C_{\idx}\le\Delta$, and that the
    cost oracle \eqref{e:cost_oracle_main} replaces the entry oracle. Then the
    same bounds hold with cost queries in place of kernel queries, and
    $\ell_K+\log(K_{\max}/K_{\min})$ may be replaced by
    $b_C+b_\eta+\log(1+\Delta/\eta)$ in $b_{\rm marg}$.
\end{lemma}
\begin{proof}
    Rescale $K$ by the known upper bound $K_{\max}$. This adds the same constant $-\log K_{\max}$ to all log-sums and does not change any normalized marginal $q_k(\beta)$. Thus the entries used by Lemma~\ref{lem:logsumexp} lie in $[\gamma,1]$ with $\gamma=K_{\min}/K_{\max}$, and the common shift cancels when a fiber log-sum is subtracted from the full tensor log-sum. We invoke Lemma~\ref{lem:logsumexp} with $r=1$ and exponent bound $mB$, since $\left|\sum_\ell\beta_{\ell,i_\ell}\right|\le mB$.

    For each $a\in[n]$, the fiber defining $S_{k,a}(\beta)$ has $n^{m-1}$ terms. Applying Lemma~\ref{lem:logsumexp} with $r=1$, accuracy $\alpha/4$, and failure probability $p/(2n)$ gives an $\alpha/4$-accurate estimate of $\log S_{k,a}(\beta)$ using $\widetilde O(n^{(m-1)/2}\alpha^{-1})$ coherent evaluations, each of which makes one kernel query and $m$ potential lookups. Applying the same primitive once to the full tensor with $r=1$ and failure probability $p/2$ gives an $\alpha/4$-accurate estimate of $\log Z(\beta)$ using $\widetilde O(n^{m/2}\alpha^{-1})$ coherent evaluations. The $n$ fiber estimates dominate both resource counts. Subtracting the estimate of $\log Z(\beta)$ from the fiber estimates gives every $\log q_{k,a}(\beta)$ to additive error at most $\alpha/2$, and hence at most $\alpha$. A union bound over the $n$ fiber estimates and the full-tensor estimate gives success probability at least $1-p$. Substitution in \eqref{e:logsumexp_word_length}, where $\log d\le m\log n$, gives $b_{\rm marg}$ and the stated gate bound.

    To compute an update correction directly, invoke the fiber call with
    $r=\mu_{k,a}$. It returns an estimate of
    $\log(S_{k,a}(\beta)/\mu_{k,a})$; subtracting the full-tensor estimate
    gives $\log q_{k,a}(\beta)-\log\mu_{k,a}$. Thus no separately rounded
    logarithm of $\mu_{k,a}$ is needed. Only $b_{\rm marg}$ changes, through
    the exact encoding of $r$ and the term $\log(1/r)$.

    For the log-domain form, invoke Lemma~\ref{lem:logsumexp} with
    $a\equiv1$, so that $\gamma=1$ and no rescaling is needed, and with the
    exact exponents $w_{\idx}(\beta)$ of Lemma~\ref{lem:cost_to_gibbs}. They
    satisfy $|w_{\idx}(\beta)|\le mB+\Delta/\eta$ and use
    $O(b_C+b_\eta+\ell_\beta+\log m)$ bits, so \eqref{e:lse_output_width}
    and \eqref{e:logsumexp_word_length} give the stated replacement. Each
    coherent evaluation then makes $O(1)$ cost queries and $m$ potential
    lookups.
\end{proof}

For the descent analysis of \eqref{e:exact_random_sinkhorn_update}, set
\[
    \mathsf{KL}_k(\beta)
    \teq
    \textsf{KL}\!\left(\mu_k\middle\|q_k(\beta)\right)
    =\sum_a\mu_{k,a}\log\frac{\mu_{k,a}}{q_{k,a}(\beta)}.
\]

\begin{lemma}[Exact and noisy one-mode Sinkhorn descent]\label{lem:noisy_sinkhorn_descent}
    The exact update \eqref{e:exact_random_sinkhorn_update} makes the selected marginal equal to the corresponding target marginal and satisfies
    \begin{equation}\label{e:exact_coordinate_decrease}
        \Phi_K(\beta)-\Phi_K(\beta^+)=\mathsf{KL}_k(\beta).
    \end{equation}
    Now suppose instead that the update uses estimates
    \[
        \widetilde\ell_a=\log q_{k,a}(\beta)+e_a,
        \qquad |e_a|\le\alpha,
    \]
    and sets $\widetilde\beta_{k,a}^{+}=\beta_{k,a}+\log\mu_{k,a}-\widetilde\ell_a$, leaving the other blocks fixed. Then
    \[
        \Phi_K(\beta)-\Phi_K(\widetilde\beta^+)
        \ge
        \mathsf{KL}_k(\beta)-\frac{\alpha^2}{2}.
    \]
\end{lemma}
\begin{proof}
    Let $Z=Z(\beta)$ and $F_a=S_{k,a}(\beta)=Zq_{k,a}(\beta)$. Under the exact update, the new unnormalized fiber mass is
    \[
        F_a^+=F_a\frac{\mu_{k,a}}{q_{k,a}(\beta)}=Z\mu_{k,a}.
    \]
    Therefore $Z(\beta^+)=\sum_aF_a^+=Z$ and the normalized $k$-th marginal is exactly $\mu_k$. Since only block $k$ changes,
    \[
        \Phi_K(\beta)-\Phi_K(\beta^+)
        =\sum_a\mu_{k,a}(\beta_{k,a}^+-\beta_{k,a})
        =\sum_a\mu_{k,a}\log\frac{\mu_{k,a}}{q_{k,a}(\beta)},
    \]
    proving \eqref{e:exact_coordinate_decrease}.

    For the noisy update, the new fiber mass is $\widetilde F_a^+=Z\mu_{k,a}e^{-e_a}$ and hence
    \[
        Z(\widetilde\beta^+)=Z\sum_a\mu_{k,a}e^{-e_a}.
    \]
    A direct substitution in $\Phi_K$ gives
    \[
        \Phi_K(\beta)-\Phi_K(\widetilde\beta^+)
        =
        \mathsf{KL}_k(\beta)-\sum_a\mu_{k,a}e_a
        -\log\!\left(\sum_a\mu_{k,a}e^{-e_a}\right).
    \]
    Hoeffding's lemma applied to the random variable $-e_a$, with $a\sim\mu_k$ and range length at most $2\alpha$, gives
    \[
        \log\!\left(\sum_a\mu_{k,a}e^{-e_a}\right)
        \le
        -\sum_a\mu_{k,a}e_a+\frac{\alpha^2}{2}.
    \]
    The claimed noisy descent bound follows.
\end{proof}

Define
\[
    K_{\min}\teq\min_{\idx}K_{\idx},
    \qquad
    K_{\max}\teq\max_{\idx}K_{\idx},
    \qquad
    \mu_{\min}\teq\min_{k,a}\mu_{k,a},
\]
and
\[
    R_K\teq \log\frac{K_{\max}}{K_{\min}}+\log\frac{1}{\mu_{\min}}.
\]
For finite control flow we use the represented integer
\begin{equation}\label{e:represented_sinkhorn_radius}
    \widehat R_K
    \teq
    1+\left\lceil\log_2\frac{K_{\max}}{K_{\min}}\right\rceil
      +\left\lceil\log_2\frac1{\mu_{\min}}\right\rceil,
\end{equation}
which is computable from the rational bounds and satisfies
$\widehat R_K\ge\max\{1,R_K\}$.
The next estimates adapt the range, objective-gap, and progress arguments behind multimarginal Sinkhorn to noisy log-marginal estimates and uniform random block selection. For the exact update, averaging \eqref{e:exact_coordinate_decrease} over a uniform random mode and applying Pinsker's inequality gives
\begin{equation}\label{e:random_coordinate_decrease_exact}
    \mathbb{E}_{k\sim\mathrm{Unif}([m])}
    \left[\Phi_K(\beta)-\Phi_K(\beta^+)\right]
    =\frac1m\sum_{k=1}^m\mathsf{KL}_k(\beta)
    \ge \frac{E(\beta)^2}{2m^2}.
\end{equation}

The log-sum-exp primitive also needs a bound on the size of the potentials it receives. The following range bound is the noisy analogue of the range induction in \cite{lin2022complexity}. Block ranges are invariant under adding a constant to a block; in the finite implementation we recenter a block by subtracting its represented minimum, which preserves the common fixed-point grid.

\begin{lemma}[Noisy Sinkhorn potential range]\label{lem:noisy_sinkhorn_range}
    Suppose that an update in mode $k$ is performed with errors $|e_a|\le\alpha$ as in Lemma~\ref{lem:noisy_sinkhorn_descent}. After recentering the updated block by an additive constant, the updated potentials satisfy
    \[
        \max_a\beta_{\ell,a}-\min_a\beta_{\ell,a}\le R_K+2\alpha,
        \qquad \ell\in[m],
    \]
    provided the same bound held before the update.
\end{lemma}
\begin{proof}
    Adding constants to the blocks changes $K(\beta)$ by a global multiplicative factor and hence leaves $\widehat K(\beta)$, all marginals, and $\Phi_K(\beta)$ unchanged. It therefore suffices to control block ranges.

    Only block $k$ changes. Write
    \[
        S_{k,a}(\beta)=e^{\beta_{k,a}}T_a,
        \qquad
        T_a\teq \sum_{\idx:i_k=a}K_{\idx}\exp\!\left(\sum_{\ell\ne k}\beta_{\ell,i_\ell}\right).
    \]
    The noisy update gives
    \[
        \widetilde\beta^+_{k,a}
        =\log\mu_{k,a}-\log T_a+\log Z(\beta)-e_a.
    \]
    For any $a,a'$, the two sums $T_a$ and $T_{a'}$ have the same index set in the remaining $m-1$ modes and termwise weights that differ only through $K$. Hence
    \[
        \frac{K_{\min}}{K_{\max}}\le \frac{T_a}{T_{a'}}\le \frac{K_{\max}}{K_{\min}}.
    \]
    Also $\mu_{k,a}/\mu_{k,a'}\le1/\mu_{\min}$. Therefore
    \[
        |\widetilde\beta^+_{k,a}-\widetilde\beta^+_{k,a'}|
        \le \log\frac{K_{\max}}{K_{\min}}+\log\frac1{\mu_{\min}}+2\alpha
        =R_K+2\alpha.
    \]
    The other blocks change only by recentering constants, so the claimed range bound is preserved.
\end{proof}

\begin{lemma}[Optimal entropic potential range]\label{lem:optimal_sinkhorn_range}
    There is a minimizer $\beta^\star$ of $\Phi_K$ whose blocks satisfy
    \[
        \max_a\beta^\star_{k,a}-\min_a\beta^\star_{k,a}\le R_K,
        \qquad k\in[m].
    \]
\end{lemma}
\begin{proof}
    At an optimum, the normalized tensor $\widehat K(\beta^\star)$ has marginals $\mu_k$, so for each block $k$ and coordinate $a$,
    \[
        \mu_{k,a}
        =\frac{e^{\beta^\star_{k,a}}T_a}{Z(\beta^\star)},
        \qquad
        T_a\teq\sum_{\idx:i_k=a}K_{\idx}\exp\!\left(\sum_{\ell\ne k}\beta^\star_{\ell,i_\ell}\right).
    \]
    For any $a,a'$, the sums $T_a$ and $T_{a'}$ have the same index set in the remaining modes and termwise weights that differ only through $K$, so
    \[
        \frac{K_{\min}}{K_{\max}}\le\frac{T_a}{T_{a'}}\le\frac{K_{\max}}{K_{\min}}.
    \]
    Also $\mu_{k,a}/\mu_{k,a'}\in[\mu_{\min},\mu_{\min}^{-1}]$. Therefore
    \[
        |\beta^\star_{k,a}-\beta^\star_{k,a'}|
        \le \log\frac{K_{\max}}{K_{\min}}+
        \log\frac1{\mu_{\min}}=R_K.
    \]
    The objective is invariant under adding constants to the individual blocks, so we may choose any centered representative of this optimal equivalence class.
\end{proof}

\begin{lemma}[Robust Sinkhorn objective gap]\label{lem:robust_sinkhorn_gap}
    Let $\beta$ be any potentials whose blocks satisfy
    \[
        \max_a\beta_{k,a}-\min_a\beta_{k,a}\le R_K+2\alpha,
        \qquad k\in[m].
    \]
    Then
    \[
        \Phi_K(\beta)-\Phi_K^\star\le (R_K+\alpha)E(\beta).
    \]
\end{lemma}
\begin{proof}
    By Lemma~\ref{lem:optimal_sinkhorn_range}, there is an optimizer $\beta^\star$ whose block ranges are at most $R_K$. Convexity gives
    \[
        \Phi_K(\beta)-\Phi_K^\star
        \le \sum_{k=1}^m(\beta_k-\beta_k^\star)^\top(q_k(\beta)-\mu_k).
    \]
    Since both $q_k(\beta)$ and $\mu_k$ have mass one, each vector $q_k(\beta)-\mu_k$ has zero sum. Thus we may subtract the midpoint of each block of $\beta$ and $\beta^\star$ inside the inner products. The centered $\beta_k$ has $\ell_\infty$ norm at most $(R_K+2\alpha)/2$, and the centered $\beta_k^\star$ has $\ell_\infty$ norm at most $R_K/2$. Therefore
    \[
        (\beta_k-\beta_k^\star)^\top(q_k(\beta)-\mu_k)
        \le (R_K+\alpha)\|q_k(\beta)-\mu_k\|_1.
    \]
    Summing over $k$ proves the claim.
\end{proof}

\begin{lemma}[Stochastic drift and occupation estimate]\label{lem:stochastic_drift_sinkhorn}
    Let $(G_t)_{0\le t\le T}$ be a bounded nonnegative process adapted to a filtration $(\mathcal F_t)_{0\le t\le T}$, where $T\ge1$ is an integer, and let $B_t\in\mathcal F_t$ be events. Let $R,\delta,c>0$ and $\chi\ge0$, and define
    \[
        h(s)\teq c\max\!\left\{\frac{\delta^2}{m^2},\frac{s^2}{m^2R^2}\right\}.
    \]
    Suppose that, for every $0\le t<T$,
    \[
        \mathbb{E}[G_t-G_{t+1}\mid\mathcal{F}_t]\ge h(G_t)
        \quad\text{on }B_t,
        \qquad
        \mathbb{E}[G_{t+1}-G_t\mid\mathcal{F}_t]\le\frac{\chi c\delta^2}{m^2}
        \quad\text{on }B_t^{c}.
    \]
    Then
    \[
        \mathbb{E}\Bigl[\sum_{t=0}^{T-1}\mathbbm 1_{B_t}\Bigr]
        \le\frac{2m^2R}{c\delta}+\chi T.
    \]
\end{lemma}
\begin{proof}
    Let $\Psi(u)\teq\int_0^u ds/h(s)$ for $u\ge0$. Since $h$ is positive and nondecreasing, $\Psi$ is nonnegative, increasing, and concave, and $\Psi'\le m^2/(c\delta^2)$. Moreover,
    \[
        \sup_{u\ge0}\Psi(u)
        \le\int_0^{R\delta}\frac{m^2}{c\delta^2}\,ds
        +\int_{R\delta}^{\infty}\frac{m^2R^2}{cs^2}\,ds
        =\frac{2m^2R}{c\delta}.
    \]
    On $B_t$, the drift assumption gives $\mathbb{E}[G_{t+1}\mid\mathcal{F}_t]\le G_t-h(G_t)$, and $G_t-h(G_t)\ge0$ because $G_{t+1}\ge0$. Jensen's inequality, monotonicity of $\Psi$, and $h(s)\le h(G_t)$ for $s\le G_t$ give
    \[
        \mathbb{E}[\Psi(G_{t+1})\mid\mathcal{F}_t]
        \le \Psi(G_t-h(G_t))
        \le \Psi(G_t)-1.
    \]
    On $B_t^c$, the same argument and the derivative bound give
    \[
        \mathbb{E}[\Psi(G_{t+1})\mid\mathcal{F}_t]
        \le \Psi\!\left(G_t+\frac{\chi c\delta^2}{m^2}\right)
        \le \Psi(G_t)+\chi.
    \]
    Hence $\mathbb{E}[\Psi(G_{t+1})-\Psi(G_t)\mid\mathcal F_t]\le\chi-\mathbbm 1_{B_t}$ for every $t$. Summing over $t<T$ and taking expectations yields
    \[
        \mathbb{E}\Bigl[\sum_{t=0}^{T-1}\mathbbm 1_{B_t}\Bigr]
        \le \Psi(G_0)-\mathbb E[\Psi(G_T)]+\chi T
        \le\frac{2m^2R}{c\delta}+\chi T.
    \]
\end{proof}

By Lemma~\ref{lem:quantum_tensor_marginals}, estimating all $n$ coordinates of one marginal to log-accuracy $\alpha$ uses $\widetilde O(n^{(m+1)/2}\alpha^{-1})$ kernel queries. This square-root fiber cost is the source of the quantum saving, but the estimates are noisy and depend on all preceding estimates, so the exact identity \eqref{e:exact_coordinate_decrease} does not by itself prove convergence. Lemma~\ref{lem:noisy_sinkhorn_descent} and \eqref{e:random_coordinate_decrease_exact} show instead that an $\alpha$-accurate update in a uniformly random mode decreases the dual in expectation by at least $E(\beta)^2/(2m^2)-\alpha^2/2$, which is positive whenever $E(\beta)>m\alpha$. Rather than certifying every iterate, which would cost $m$ marginal estimates per iteration, we run a fixed number of updates and output the potentials at a uniformly random iteration; Lemma~\ref{lem:stochastic_drift_sinkhorn} shows that most iterates have small residual.

\begin{theorem}[Quantum randomized multimarginal Sinkhorn]\label{t:quantum_tensor_sinkhorn}
    Let $K\in\mathbb{Q}_{>0}^{[n]^m}$ and
    $\mu_1,\ldots,\mu_m\in\Delta^n_{>0}\cap\mathbb Q^n$ have exact
    encodings of at most $\ell_K$ and $\ell_\mu$ bits per entry, and let
    $0<K_{\min}\le K_{\idx}\le K_{\max}$ be known rational bounds. Assume
    exact coherent entry access to $K$ and the marginals. For every represented
    $\delta\in(0,1)$, there is a quantum algorithm which, with probability at
    least $8/9$, outputs fixed-point potentials
    $\beta=(\beta_1,\ldots,\beta_m)$ such that
    \[
        \sum_{k=1}^m\|\proj_k(\widehat K(\beta))-\mu_k\|_1\le \delta.
    \]
    Let $\ell_\delta$ be the encoding length of $\delta$. The algorithm uses
    \[
        \widetilde O\!\left(
            m^3n^{\frac{m+1}{2}}\frac{\widehat R_K}{\delta^2}
        \right)
    \]
    coherent kernel queries and
    $\widetilde O(m^4n^{(m+1)/2}\widehat R_K\delta^{-2})$ working-vector
    lookups and unit-cost arithmetic operations. If
    \[
        b_{\rm work}\teq O\!\left(
            \ell_K+\ell_\mu+\ell_\delta+\widehat R_K
            +m\log n+\log\frac{m\widehat R_K}{\delta}
        \right),
    \]
    then its non-oracle elementary-gate complexity, excluding implementation
    of the kernel and lookup oracles, is
    $\widetilde O(m^4n^{(m+1)/2}\widehat R_K\delta^{-2}\mathsf{Mult}(b_{\rm work}))$.
    Under the log-domain access of Lemma~\ref{lem:quantum_tensor_marginals},
    where $K_{\idx}=e^{-C_{\idx}/\eta}$ need not be rational and any computable
    integer at least $\max\{1,R_K\}$ may replace $\widehat R_K$, the same
    bounds hold with cost queries in place of kernel queries and with
    $b_{\rm work}$ replaced by
    \[
        b^{\log}_{\rm work}\teq O\!\left(
            b_C+b_\eta+\ell_\mu+\ell_\delta
            +m\log n+\log\frac{m\widehat R_K}{\delta}
        \right).
    \]
\end{theorem}

\begin{algorithm}[t]
\caption{Noisy quantum randomized multimarginal Sinkhorn}
\label{alg:qsink}
\begin{algorithmic}[1]
\Require Coherent access to $K>0$, known $K_{\min},K_{\max}$, marginals $\musf>0$, $\delta\in(0,1)$
\State $\beta\gets0$, $\alpha\gets\delta/(32m)$, $T\gets\lceil1536\,m^2\widehat R_K/\delta\rceil$
\State Sample the output time $t^\star\sim\mathrm{Unif}(\{0,\ldots,T-1\})$
\For{$t=0,\ldots,t^\star-1$}
    \State Sample $k\sim\mathrm{Unif}([m])$
    \State Estimate $\widetilde d_a\approx\log q_{k,a}(\beta)-\log\mu_{k,a}$ to total accuracy $\alpha$ for every $a\in[n]$
    \State $\beta_{k,a}\gets\beta_{k,a}-\widetilde d_a$ for every $a\in[n]$
    \State Recenter block $\beta_k$ by subtracting $\min_a\beta_{k,a}$
    \If{some block range exceeds $\widehat R_K+2\alpha$} \State \Return \textsc{Failure} \EndIf
\EndFor
\State \Return $\beta$
\end{algorithmic}
\end{algorithm}

\begin{proof}
    Set $\alpha\teq\delta/(32m)$, $T\teq\lceil1536\,m^2\widehat R_K/\delta\rceil$, and $p\teq1/(48T)$. Before any update, the algorithm samples $t^\star$ uniformly from $\{0,\ldots,T-1\}$, independently of everything else. Starting from $\beta^{(0)}=0$, for $t=0,\ldots,t^\star-1$ it samples $k_t$ uniformly from $[m]$ and uses the target-mass assertion of Lemma~\ref{lem:quantum_tensor_marginals}, with $r=\mu_{k_t,a}$, to produce fixed-point corrections
    \[
        \widetilde d_a
        =\log q_{k_t,a}(\beta^{(t)})-\log\mu_{k_t,a}+e_a,
        \qquad |e_a|\le\alpha.
    \]
    It performs the noisy update $\beta_{k_t,a}\gets\beta_{k_t,a}-\widetilde d_a$ and recenters the updated block by subtracting its represented minimum. The subtraction and recentering are exact operations on the common fixed-point grid, so $e_a$ contains the full approximation error of the implemented update. Recentering multiplies $K(\beta)$ by a global constant and therefore leaves $\widehat K(\beta)$, all marginals, and $\Phi_K$ unchanged. The output is $\beta^{(t^\star)}$; no residual certificate is computed.

    Every initial block has range zero, so the hypothesis of Lemma~\ref{lem:noisy_sinkhorn_range} holds at $t=0$ and is preserved by every update whose estimates are $\alpha$-accurate. After recentering, the algorithm checks in $O(M)$ time that every block range is at most $\widehat R_K+2\alpha$ and aborts otherwise; since $R_K\le\widehat R_K$, Lemma~\ref{lem:noisy_sinkhorn_range} shows that it can abort only after an estimator failure. Every represented block then lies in $[0,\widehat R_K+2\alpha]$, so every quantum call receives potentials with $\|\beta\|_\infty=O(\widehat R_K)$, as required for the worst-case cost, regardless of earlier estimator outputs. Choosing the output widths according to \eqref{e:lse_output_width} prevents overflow, while $b_F=O(\log(m/\delta))$ suffices for the requested update accuracy.

    To avoid conditioning on estimator success, we analyze an auxiliary safe process. It runs all $T$ updates with the same random choices and, whenever an estimator call fails, replaces the estimate by the exact correction rounded to the estimator's fixed-point output grid. We choose that grid with spacing at most $\alpha/4$, so the replacement has error at most $\alpha$ and remains a valid represented input to every later estimator call. Every safe update therefore satisfies Lemma~\ref{lem:noisy_sinkhorn_descent}, and every safe iterate satisfies the range hypothesis of Lemma~\ref{lem:robust_sinkhorn_gap}. Each of the at most $T$ estimator calls fails with conditional probability at most $p$ given the complete past, so the tower property and a union bound show that all calls succeed except with probability at most $Tp=1/48$. On that event, the executable and safe trajectories agree through time $t^\star$.

    Let $\mathcal F_t$ be the history of the safe process through time $t$, let $G_t\teq\Phi_K(\beta^{(t)})-\Phi_K^\star$, and call $t$ bad if $E(\beta^{(t)})>\delta/2$. Lemma~\ref{lem:robust_sinkhorn_gap}, $E\le2m$, and $R_K+\alpha\le2\widehat R_K$ give $0\le G_t\le4m\widehat R_K$. Since $k_t$ is uniform given $\mathcal F_t$, Lemma~\ref{lem:noisy_sinkhorn_descent} and \eqref{e:random_coordinate_decrease_exact} give
    \begin{equation}\label{e:safe_sinkhorn_drift}
        \mathbb E[G_t-G_{t+1}\mid\mathcal F_t]
        \ge\frac1m\sum_{k=1}^m\mathsf{KL}_k(\beta^{(t)})-\frac{\alpha^2}2
        \ge\frac{E(\beta^{(t)})^2}{2m^2}-\frac{\delta^2}{2048m^2}.
    \end{equation}
    If $t$ is bad, then $\delta^2<4E(\beta^{(t)})^2$, so the right-hand side is at least $E(\beta^{(t)})^2/(4m^2)$. Lemma~\ref{lem:robust_sinkhorn_gap} also gives $E(\beta^{(t)})\ge G_t/(2\widehat R_K)$. With $R\teq2\widehat R_K$ and $c\teq1/16$, every bad $t$ therefore satisfies
    \[
        \mathbb E[G_t-G_{t+1}\mid\mathcal F_t]
        \ge c\max\!\left\{\frac{\delta^2}{m^2},\frac{G_t^2}{m^2R^2}\right\}.
    \]
    If $t$ is not bad, \eqref{e:safe_sinkhorn_drift} and $\mathsf{KL}_k\ge0$ give $\mathbb E[G_{t+1}-G_t\mid\mathcal F_t]\le\delta^2/(2048m^2)=\chi c\delta^2/m^2$ with $\chi\teq1/128$. Lemma~\ref{lem:stochastic_drift_sinkhorn} therefore bounds the expected number of bad times in $\{0,\ldots,T-1\}$ by
    \[
        \frac{2m^2R}{c\delta}+\chi T
        =\frac{64m^2\widehat R_K}{\delta}+\frac T{128}
        \le\left(\frac1{24}+\frac1{128}\right)T.
    \]
    Because $t^\star$ is uniform on $\{0,\ldots,T-1\}$ and independent of the safe process, $\Pr[E(\beta^{(t^\star)})>\delta/2]$ equals this expectation divided by $T$. A union bound with the estimator-failure event shows that the executable output satisfies $E(\beta)\le\delta/2$ except with probability at most $1/24+1/128+1/48<1/9$.

    Each update calls Lemma~\ref{lem:quantum_tensor_marginals} once at accuracy $\alpha=\Theta(\delta/m)$, which costs $\widetilde O(mn^{(m+1)/2}/\delta)$ kernel queries and $\widetilde O(m^2n^{(m+1)/2}/\delta)$ potential lookups. At most $T=O(m^2\widehat R_K/\delta)$ updates are performed, which yields $\widetilde O(m^3n^{(m+1)/2}\widehat R_K/\delta^2)$ kernel queries and $\widetilde O(m^4n^{(m+1)/2}\widehat R_K/\delta^2)$ lookups. Substituting the exact input lengths, $p=1/(48T)$, the range bound on $\beta$, and $\mu_{\min}$ into $b_{\rm marg}$, or into its log-domain form, gives $b_{\rm work}$ or $b^{\log}_{\rm work}$ and the stated gate complexity.
\end{proof}

We now specialize to entropy-regularized MOT. For $K_{\idx}=\exp(-C_{\idx}/\eta)$, scaling $K$ to the target marginals gives the optimal solution to \eqref{e:entropic_mot_tensor}.
\begin{corollary}[Quantum entropic MOT via multimarginal Sinkhorn]\label{c:quantum_entropic_mot_tensor}
    Suppose $n_1=\cdots=n_m=n$, the nonnegative cost entries and positive
    marginals have exact encodings of at most $b_C$ and $b_\mu$ bits per
    entry, and there is coherent cost access through
    \eqref{e:cost_oracle_main} and coherent marginal lookup. Let
    $\eta=2^{-e_\eta}$ for an integer $e_\eta$, write
    $b_\eta\teq|e_\eta|+1$, and let $\delta\in(0,1)$ have encoding length
    $b_\delta$.
    There is a quantum algorithm which outputs fixed-point scaling potentials
    specifying the exact mathematical Gibbs tensor
    $X_{\idx}\propto e^{-C_{\idx}/\eta+\sum_k\beta_{k,i_k}}$ satisfying
    \[
        \sum_{k=1}^m\|\proj_k(X)-\mu_k\|_1\le\delta
    \]
    with probability at least $8/9$, using
    \[
        \widetilde O\!\left(m^3 n^{\frac{m+1}{2}}
            \widehat R_{C,\mu,\eta}\delta^{-2}\right)
    \]
    coherent cost queries and
    $\widetilde O(m^4n^{(m+1)/2}\widehat R_{C,\mu,\eta}\delta^{-2})$
    working-vector lookups and unit-cost arithmetic operations, where
    \[
        \tau_\mu\teq\log\frac{1}{\min_{k,a}\mu_{k,a}},
        \qquad
        \widehat R_{C,\mu,\eta}
        \teq 2+
        \left\lceil\frac{\|C\|_\infty}{\eta\log2}\right\rceil
        +\left\lceil\log_2\frac1{\min_{k,a}\mu_{k,a}}\right\rceil.
    \]
    Set
    \[
        b_{\rm ent}\teq O\!\left(
            b_C+b_\mu+b_\eta+b_\delta
            +m\log n+\log\frac{m\widehat R_{C,\mu,\eta}}{\delta}
        \right).
    \]
    Its non-oracle elementary-gate complexity is
    \[
        \widetilde O\!\left(
            m^4n^{\frac{m+1}{2}}
            \widehat R_{C,\mu,\eta}\delta^{-2}
            \mathsf{Mult}(b_{\rm ent})
        \right).
    \]
\end{corollary}
\begin{proof}
    Let $K_{\idx}=e^{-C_{\idx}/\eta}$ denote the exact mathematical kernel.
    Since $C\ge0$, $K_{\max}\le1$, $K_{\min}\ge
    e^{-\|C\|_\infty/\eta}$, and hence
    $\widehat R_K\le\widehat R_{C,\mu,\eta}$; the algorithm uses the
    computable integer $\widehat R_{C,\mu,\eta}$ in place of $\widehat R_K$.
    The circuit never evaluates the generally transcendental number
    $K_{\idx}$. By Lemma~\ref{lem:cost_to_gibbs}, every Gibbs exponent
    $\sum_k\beta_{k,i_k}-C_{\idx}/\eta$ is computed exactly from $O(1)$ cost
    queries and $m$ potential lookups, so the log-domain form of
    Lemma~\ref{lem:quantum_tensor_marginals}, with $\Delta=\|C\|_\infty$ and
    $r=\mu_{k,a}$, returns corrections whose errors are measured with respect
    to the exact kernel. Therefore the noisy-descent, range, and
    stochastic-drift proof of Theorem~\ref{t:quantum_tensor_sinkhorn}
    applies verbatim to the exact Gibbs tensor, giving the stated residual
    and query bounds. Substituting $\alpha=\Theta(\delta/m)$,
    $p^{-1}=\widetilde O(m^2\widehat R_{C,\mu,\eta}/\delta)$, and
    $\log(1+\|C\|_\infty/\eta)=O(\log\widehat R_{C,\mu,\eta})$ into
    $b^{\log}_{\rm work}$ gives $b_{\rm ent}$ and the stated gate bound.
\end{proof}

The next lemma bounds the transport cost of a nearly feasible Gibbs tensor. It is used both for dense rounding below and for sparse recovery in Section~\ref{s:sparse_recovery}.
\begin{lemma}[Cost of a nearly feasible Gibbs tensor]\label{lem:gibbs_pre_recovery}
Let $C\in[0,\Delta_C]^{[n]^m}$, let $\eta>0$, and let $Q$ be the normalized Gibbs tensor of some potentials $\beta$, that is, $Q_{\idx}\propto e^{-C_{\idx}/\eta+\sum_k\beta_{k,i_k}}$. Write $E_{\musf}(Q)\teq\sum_k\|\proj_k(Q)-\mu_k\|_1$. Then
\begin{equation}\label{e:gibbs_pre_recovery_cost_main}
    \langle C,Q\rangle
    \le \mathsf{MOT}(C,\musf)+2\Delta_C E_{\musf}(Q)+\eta m\log n.
\end{equation}
\end{lemma}
\begin{proof}
Let $\nu_k\teq\proj_k(Q)$. The entropic problem \eqref{e:entropic_mot_tensor} with target marginals $\boldsymbol\nu=(\nu_1,\ldots,\nu_m)$ is strictly convex, and $Q$ is feasible for it and has the Gibbs form \eqref{e:entropic_mot_scaling_form}; the KKT conditions are therefore satisfied, so $Q$ is its exact optimizer. Let $X^\star$ be an optimal coupling for the target marginals $\musf$. The rounding theorem of Lin et al.~\cite[Theorem~4.4]{lin2022complexity}, applied to $X^\star$ with target marginals $\boldsymbol\nu$, gives a coupling $Y$ with marginals $\boldsymbol\nu$ and $\|Y-X^\star\|_1\le2E_{\musf}(Q)$. Entropic optimality of $Q$, nonnegativity of entropy, and $H(Q)\le m\log n$ therefore give
\[
    \langle C,Q\rangle
    \le\langle C,Y\rangle+\eta H(Q)-\eta H(Y)
    \le\langle C,X^\star\rangle+\Delta_C\|Y-X^\star\|_1+\eta m\log n,
\]
which is \eqref{e:gibbs_pre_recovery_cost_main}.
\end{proof}

Accounting for the entropic bias and the multimarginal rounding argument of Lin et al.~\cite[Theorems~4.4--4.5]{lin2022complexity} gives the following explicit result for unregularized MOT. The returned potentials describe an approximately marginal-feasible Gibbs tensor; they do not themselves encode the final feasible tensor produced by dense rounding.
\begin{theorem}[Quantum multimarginal Sinkhorn for MOT]\label{t:quantum_sinkhorn_full}
Let $(C,\musf,\varepsilon)$ be a represented instance in the sense of
Section~\ref{s:finite_input_appendix}, and define
$\kappa=\max\{1,(\max C-\min C)/\varepsilon\}$. Under coherent cost access
\eqref{e:cost_oracle_main}, marginal lookup, and coherent lookup of the length-$mn$ potentials,
a quantum algorithm computes, with probability at least $2/3$, fixed-point
potentials whose exact mathematical Gibbs tensor can be rounded to an exactly
feasible additive-$\varepsilon$ coupling. Computing the range and the
potentials uses $\widetilde O(m^4n^{(m+1)/2}\kappa^3)$ coherent cost queries
and $\widetilde O(m^5n^{(m+1)/2}\kappa^3)$ working-vector lookups and unit-cost
arithmetic operations. Set
\[
    b_{\rm sink}\teq O\!\left(b_C+b_\mu+b_\varepsilon+m\log n\right).
\]
The non-oracle elementary-gate complexity under quadratic reversible
fixed-point arithmetic is
\[
    \widetilde O\!\left(
        m^5n^{\frac{m+1}{2}}\kappa^3b_{\rm sink}^2
    \right).
\]
In particular, this is
$\widetilde O(m^7n^{(m+1)/2}\kappa^3)$ when
$b_C+b_\mu+b_\varepsilon=O(m\log(n\kappa))$. Let $T_C$, $T_W$, and $T_U$
denote the gate costs of one coherent cost query, one working-vector lookup,
and one working-vector write; Remark~\ref{rem:no_qram} gives table-based
values. A compiled implementation adds
$\widetilde O(m^4n^{\frac{m+1}{2}}\kappa^3T_C+m^5n^{\frac{m+1}{2}}\kappa^3T_W
    +m^3n\kappa^2T_U)$
gates for coherent reads and writes.
\end{theorem}
\begin{proof}
    Compute $L=\min C$ and $U=\max C$ by quantum minimum and maximum finding~\cite{durr1996quantum,van2020quantum}, jointly boosted to failure probability at most $1/6$; this uses $\widetilde O(\sqrt N)$ cost queries, which is lower order. Then $\Delta_C=U-L$. If $\Delta_C\le\varepsilon$, set $\eta=1$ and return the fixed-point potentials $\beta=0$. Their exact Gibbs tensor can be rounded to the target marginals, and every resulting feasible coupling has cost at most $U\le\mathsf{MOT}(C,\musf)+\varepsilon$. Assume henceforth that $\Delta_C>\varepsilon$, so $\kappa=\Delta_C/\varepsilon$.

    Replace $C$ by $C^+\teq C-L\mathbf1\in[0,\Delta_C]^{[n]^m}$. The shift changes every feasible objective by $L$ and multiplies the Gibbs kernel by a global constant, so it preserves additive optimality gaps, normalized Gibbs tensors, and scaling potentials. Set
    \[
        \begin{aligned}
        \eta&\teq\operatorname{dyad}\!\left(\frac{\varepsilon}{4mq_n}\right),
        &\delta&\teq\operatorname{dyad}\!\left(\frac{\varepsilon}{4\Delta_C}\right),\\
        \lambda&\teq\operatorname{dyad}\!\left(\frac{\delta}{8m}\right),
        &\mu'_k&\teq(1-\lambda)\mu_k+\lambda u_n^{\rm dyad}.
        \end{aligned}
    \]
    where $u_n^{\rm dyad}$ is the represented distribution defined in
    Section~\ref{s:finite_input_appendix}. Then
    \[
        \sum_{k=1}^m\|\mu'_k-\mu_k\|_1\le2m\lambda\le\frac\delta4,
        \qquad
        \min_{k,a}\mu'_{k,a}\ge\frac{\lambda}{2n}>\frac{\delta}{32mn}.
    \]
    Each $\mu'_{k,a}$ is rational with encoding length
    $O(b_\mu+b_C+b_\varepsilon+\log(mn\kappa))$, so the smoothed instance
    remains within the finite-input model.
    Run the algorithm of Corollary~\ref{c:quantum_entropic_mot_tensor} for $C^+$ and $\musf'$ with residual $\delta/4$. Except with probability at most $1/9$, its potentials define a Gibbs tensor $Q$ with $E_{\musf}(Q)\le\delta/2$, by the triangle inequality. The rounding procedure of Lin et al.~\cite[Theorem~4.4]{lin2022complexity}, applied to $Q$ written densely, returns $\widehat X\in\Ucal(\musf)$ with $\|\widehat X-Q\|_1\le2E_{\musf}(Q)$. Together with Lemma~\ref{lem:gibbs_pre_recovery} for $C^+$, this gives
    \begin{align*}
        \langle C^+,\widehat X\rangle
        &\le\langle C^+,Q\rangle+2\Delta_CE_{\musf}(Q)\\
        &\le\textup{\textsf{MOT}}(C^+,\musf)+4\Delta_CE_{\musf}(Q)+\eta m\log n
        \le\textup{\textsf{MOT}}(C^+,\musf)+\frac{3\varepsilon}4.
    \end{align*}
    Shifting back adds $L$ to both sides. A union bound over range finding and Sinkhorn gives success probability at least $1-1/6-1/9>2/3$.

    The dyadic-floor inequalities give $\Delta_C/\eta<8mq_n\kappa$ and
    $\delta^{-1}<8\kappa$. Moreover,
    $\tau_\mu\le\log(32mn/\delta)=O(\log(mn\kappa))$, so
    $\widehat R_{C,\mu,\eta}=\widetilde O(m\kappa)$. The dyadic formulas
    also give $b_\eta+b_\delta=O(b_C+b_\varepsilon+\log(mn\kappa))$;
    substituting in the definition of $b_{\rm ent}$ yields
    $b_{\rm ent}=O(b_{\rm sink})$. Corollary~\ref{c:quantum_entropic_mot_tensor} therefore uses $\widetilde O(m^3\cdot m\kappa\cdot\kappa^2\cdot n^{(m+1)/2})=\widetilde O(m^4n^{(m+1)/2}\kappa^3)$ coherent cost queries and $\widetilde O(m^5n^{(m+1)/2}\kappa^3)$ working-vector lookups and unit-cost arithmetic operations. Its gate bound, including the exact exponent computation of Lemma~\ref{lem:cost_to_gibbs}, is $\widetilde O(m^5n^{(m+1)/2}\kappa^3b_{\rm sink}^2)$. Under the word-length condition in the theorem this becomes $\widetilde O(m^7n^{(m+1)/2}\kappa^3)$. Each update writes $n$ potential entries, and there are $O(m^2\widehat R_{C,\mu,\eta}/\delta)=\widetilde O(m^3\kappa^2)$ updates, which gives the write term in the theorem. These bounds concern the scaling-potential output; materializing and rounding $Q$ costs an additional $O(mn^m)$ classical operations.
\end{proof}

\section{Sparse coupling recovery from Gibbs potentials}\label{s:sparse_recovery}
The potentials of Theorem~\ref{t:quantum_sinkhorn_full} do not yet give the requested output. The tensor $Q_\beta$ they describe is implicit, and standard rounding first materializes all $N$ entries, which would turn a sublinear quantum computation into an $\Omega(N)$ classical one. We therefore return a classical list $\{(\idx_r,w_r):r\in[s']\}$ representing $X=\sum_rw_re_{\idx_r}$. Here ``sparse'' means list-form rather than dense-tensor output; the list is sublinear in $N$ only in the regime stated after Theorem~\ref{t:quantum_recovery_full}. The list provides an $\varepsilon$-optimal coupling: $w_r\ge0$, $\proj_k(X)=\mu_k$ for every $k$, and $\langle C,X\rangle\le\mathsf{MOT}(C,\musf)+\varepsilon$. Exact feasibility holds in exact arithmetic on the input masses, because every output weight lies on a common finite rational grid.

A value or dual certificate reports only the optimum, and scaling potentials describe a Gibbs law only implicitly. After $O(s')$ preprocessing, a list supports repeated sampling, statistics $\sum_r w_r f(\idx_r)$, and arbitrary pushforwards without further cost queries. For example, pushing its atoms through a barycentric map gives a discrete Wasserstein barycenter~\cite{Agueh2011Barycenters,Carlier2015Numerical}, and in dependence models the list records a joint law on which new payoffs can be evaluated~\cite{Dolinsky2014Robust,Galichon2014Stochastic}. Because the list can be reused in this way, we treat exact sparse recovery as an algorithmic problem in its own right, as Altschuler and Boix-Adser{\`a}~\cite{altschuler2023structured} do for structured costs.

We convert the Gibbs potentials into a list describing an $\varepsilon$-optimal coupling in three steps: trim, sample, and complete. Starting from $E(\beta)\le\delta$, we estimate each marginal of $Q_\beta$ from above and scale its entries down, using one correction factor per mode. The trimmed tensor $Q^\downarrow$ lies below every target marginal coordinatewise, and the residual bound ensures that trimming removes only $O(\delta)$ mass. This one-sided invariant makes the final completion possible, because every missing marginal mass is then a nonnegative deficit.

We cannot enumerate $Q^\downarrow$, but we can sample from it.
Use the proposal $P_{\idx}=\mu_{1,i_1}/n^{m-1}$: draw the first coordinate
from $\mu_1$ and the others uniformly. The domination
$Q^\downarrow_{\idx}\le\proj_1(Q^\downarrow)_{i_1}\le\mu_{1,i_1}$ makes
rejection sampling valid. Its acceptance probability is of order
$n^{-(m-1)}$, because the proposal spreads its mass uniformly over the remaining
$m-1$ coordinates while $Q^\downarrow$ retains a constant total mass.
Amplitude amplification~\cite{brassard2000quantum} therefore produces one accepted sample in
$\widetilde O(n^{(m-1)/2})$ weight queries, independently of the
smallest marginal weight. If $N_{\rm samp}$ samples are drawn, their expected
total marginal error is $O(m\sqrt{n/N_{\rm samp}})$; choosing
$N_{\rm samp}=O(m^2n\theta^{-2}\log(m/\zeta))$ makes this error at most
$\theta$ with failure probability $\zeta$. The same samples control the
transport cost.

Sampling introduces small marginal errors, so we trim the
empirical tensor coordinate by coordinate until all of its marginals again lie
below their targets. The $m$ remaining deficit vectors are nonnegative and have
equal total mass, and a greedy procedure couples them using at most $mn-m+1$
additional atoms, which restores every marginal exactly. No step
materializes an $N$-entry tensor.

\begin{lemma}[Finite-precision amplified rejection sampling]\label{lem:finite_precision_rejection}
Let $P$ be a probability distribution on a finite set $\Omega$, let
$A_i\in[0,1]$, and define
\[
    a\teq\sum_{i\in\Omega}P_iA_i,
    \qquad
    R_i\teq\frac{P_iA_i}{a}.
\]
Suppose $a\ge a_0>0$, an ideal unitary $U_P$ prepares
$\sum_i\sqrt{P_i}\ket i$, and a reversible fixed-point circuit computes
$\widetilde A_i\in[0,1]$ with
\[
    \max_i|\widetilde A_i-A_i|\le \frac{a_0\varphi}{8}.
\]
For every $\varphi,\sigma\in(0,1/4)$, exponential-search amplitude
amplification returns, except with probability $\sigma$, a sample from a
distribution within total variation distance $\varphi$ of $R$ using
\[
    O\!\left(a_0^{-1/2}\log\frac1\sigma\right)
\]
state-preparation and acceptance-circuit invocations. It suffices to use
$O(\log(a_0^{-1}\varphi^{-1}\sigma^{-1}))$ additional fractional bits
beyond the input word lengths. The same guarantee holds for an implemented
preparation unitary $\widetilde U_P$ provided
\[
    \|\widetilde U_P-U_P\|_{\rm op}
    \le c\,\varphi\sqrt{a_0}/\log(1/\sigma)
\]
for a sufficiently small absolute constant $c>0$. Under quadratic reversible arithmetic, if
one acceptance probability uses $O(m)$ fixed-point operations on $b$-bit
words, its non-oracle gate cost is $\widetilde O(mb^2)$ per invocation.
\end{lemma}
\begin{proof}
Let $\widetilde a=\sum_iP_i\widetilde A_i$ and
$\widetilde R_i=P_i\widetilde A_i/\widetilde a$. Then
\[
    D\teq\sum_iP_i|\widetilde A_i-A_i|
    \le \frac{a_0\varphi}{8},
    \qquad
    |\widetilde a-a|\le D.
\]
Consequently,
\[
    \|R-\widetilde R\|_{\rm TV}
    \le \frac{D}{a}
    \le \frac\varphi8.
\]
Amplitude amplification preserves the normalized accepting component, so
exact implementations of state preparation and the controlled rotation for
$\widetilde A_i$ sample exactly from $\widetilde R$ when they succeed. The exponential-search procedure of
Boyer et al.~\cite{boyer1998tight}, combined with amplitude amplification~\cite{brassard2000quantum}, uses
$O(a_0^{-1/2}\log(1/\sigma))$ invocations. Implement state preparation and
the acceptance rotation in each invocation to operator-norm error at most a constant times
$\varphi a_0^{1/2}/\log(1/\sigma)$. A hybrid argument then changes the
output distribution by at most $O(\varphi)$; adjusting the constants gives
the stated total variation bound. Fixed-point rounding to absolute error
$O(a_0\varphi)$ and this rotation accuracy require the displayed number of
additional bits. Reversible multiplication, division, square root, and
rotation synthesis on $b$-bit words cost $\widetilde O(b^2)$ gates under
quadratic arithmetic, proving the final claim.
\end{proof}

\begin{lemma}[Trim--sample--complete conversion]\label{lem:sparse_gibbs_recovery_full}
Let $Q$ be a probability tensor on $[n]^m$ with $E_{\musf}(Q)=\sum_k\|\proj_k(Q)-\mu_k\|_1\le\delta\le1/8$. Suppose $Q_{\idx}=W_{\idx}/Z_Q$, with coherent fixed-point access to $\log W_{\idx}$ and the log-sum-exp access used above. Fix a nonnegative cost tensor $C$; it may be the cost that defines $W$, but it does not depend on the algorithm's randomness. For every $\theta,\zeta\in(0,1/4)$, a quantum algorithm returns, with probability at least $1-\zeta$, a classical list for an exactly feasible coupling $\widetilde X\in\Ucal(\musf)$ such that
\[
    |\operatorname{supp}(\widetilde X)|
    =O\!\left(m^2n\theta^{-2}\log\frac{m}{\zeta}+mn\right),
    \qquad
    \langle C,\widetilde X\rangle
    \le \langle C,Q\rangle
    +\left(\frac58\delta+\frac32\theta\right)\|C\|_\infty.
\]
Its sampling stage uses
$
    \widetilde O(m^2n^{\frac{m+1}{2}}\theta^{-2}\log\zeta^{-1})
$
coherent weight queries and $O(m^3n\theta^{-2}\log(m/\zeta)+mn)$ additional classical arithmetic operations, under the marginal-state access of Section~\ref{s:finite_input_appendix}.
\end{lemma}
\begin{proof}
    Write $\nu_k=\proj_k(Q)$ and set
    $\alpha\teq\operatorname{dyad}(\delta/(8m))$. Run the final
    marginal-estimation calls with log-accuracy $\alpha/2$ and failure
    probability $\zeta/(8m)$ for each mode. Let
    $\widetilde\ell_{k,a}$ be the returned fixed-point estimate and store the
    upper log-marginal
    \[
        u_{k,a}\teq\widetilde\ell_{k,a}+\alpha/2
    \]
    exactly on their common dyadic grid. Thus, except with probability at most
    $\zeta/8$, the mathematical numbers
    $\bar\nu_{k,a}\teq e^{u_{k,a}}$ satisfy
    \begin{equation}\label{e:upper_marginal_estimates}
        \nu_{k,a}\le\bar\nu_{k,a}\le e^\alpha\nu_{k,a}
        \qquad(k\in[m],\ a\in[n]).
    \end{equation}
    The algorithm retains the represented upper logs rather than exponentiating
    them. Define the ideal trimming factors equivalently by
    \[
        f_{k,a}\teq
        \begin{cases}
        0,&\mu_{k,a}=0,\\
        \exp\!\bigl(\min\{0,\log\mu_{k,a}-u_{k,a}\}\bigr),
            &\mu_{k,a}>0,
        \end{cases}
        =\min\left\{1,\frac{\mu_{k,a}}{\bar\nu_{k,a}}\right\},
        \qquad
        Q^\downarrow_{\idx}\teq Q_{\idx}\prod_{k=1}^m f_{k,i_k}.
    \]
    Since $Q^\downarrow\le f_{k,i_k}Q$ for each $k$, \eqref{e:upper_marginal_estimates} implies
    \begin{equation}\label{e:trimmed_submarginals}
        \proj_k(Q^\downarrow)_a
        \le f_{k,a}\nu_{k,a}\le\mu_{k,a}.
    \end{equation}
    Let $s=\|Q^\downarrow\|_1$ and $\tau=1-s$.  Using $1-\prod_k t_k\le\sum_k(1-t_k)$ for $t_k\in[0,1]$, we obtain
    \[
        \tau
        \le\sum_{k=1}^m\sum_{a=1}^n\nu_{k,a}(1-f_{k,a}).
    \]
    If $\bar\nu_{k,a}>\mu_{k,a}$, then
    \[
        \nu_{k,a}(1-f_{k,a})
        =\nu_{k,a}-\frac{\mu_{k,a}\nu_{k,a}}{\bar\nu_{k,a}}
        \le (\nu_{k,a}-\mu_{k,a})_+
        +\mu_{k,a}(1-e^{-\alpha});
    \]
    if $\bar\nu_{k,a}\le\mu_{k,a}$, the left-hand side is zero.  Since the positive deviations of two probability vectors have total mass one half of their $\ell_1$ distance,
    \begin{equation}\label{e:trimmed_mass_loss}
        \tau\le\frac12E_{\musf}(Q)+m(1-e^{-\alpha})
        \le\frac12\delta+m\alpha\le\frac58\delta.
    \end{equation}
    Set $s_0\teq1-5\delta/8$.  On the event above, $0<s_0\le s$.  The value $s$ is used only to describe the conditional sampling distribution; the algorithm uses the explicit scalar $s_0$ when assigning masses to the sampled atoms.

    We next sample from $R\teq Q^\downarrow/s$.  Let
    \[
        W^\downarrow_{\idx}\teq W_{\idx}\prod_{k=1}^m f_{k,i_k},
        \qquad Q^\downarrow_{\idx}=W^\downarrow_{\idx}/Z_Q.
    \]
    Estimate $\log Z_Q$ to additive accuracy $\alpha_Z\le1/16$ and failure
    probability $\zeta/8$, and store its fixed-point upper confidence endpoint
    $u_Z$. Let $\bar Z_Q\teq e^{u_Z}$ mathematically. On the resulting success
    event,
    \[
        Z_Q\le\bar Z_Q\le e^{2\alpha_Z}Z_Q.
    \]
    Prepare the proposal state
    \[
        \sum_{\idx\in[n]^m}
        \sqrt{P_{\idx}}\ket{\idx},
        \qquad
        P_{\idx}\teq\frac{\mu_{1,i_1}}{n^{m-1}},
    \]
    by preparing $\sqrt{\mu_1}$ in the first register and uniform superpositions in the remaining registers. Since $Q^\downarrow$ is nonnegative, Equation~\eqref{e:trimmed_submarginals} gives the entrywise chain
    \[
        Q^\downarrow_{\idx}
        \le \proj_1(Q^\downarrow)_{i_1}
        \le \mu_{1,i_1},
    \]
    and hence
    \[
        0\le
        A_{\idx}\teq
        \frac{W^\downarrow_{\idx}}{\bar Z_Q\mu_{1,i_1}}
        \le1
    \]
    whenever $\mu_{1,i_1}>0$.  Set $A_{\idx}=0$ when $\mu_{1,i_1}=0$; then $W^\downarrow_{\idx}=0$ as well.  Coherently compute $A_{\idx}$ and perform a controlled rotation with accepting probability $A_{\idx}$.  The total acceptance probability is
    \[
        \sum_{\idx}P_{\idx}A_{\idx}
        =\frac{sZ_Q}{n^{m-1}\bar Z_Q}
        \ge e^{-2\alpha_Z}\frac{s}{n^{m-1}}.
    \]
    Since $\delta\le1/8$, \eqref{e:trimmed_mass_loss} gives $s\ge1-5/64>1/2$. Hence the acceptance probability is at least
    $a_0\teq(3n^{m-1})^{-1}$. Apply
    Lemma~\ref{lem:finite_precision_rejection} to each sample with
    $\varphi=\sigma=\zeta/(16N_{\rm samp})$. We now verify its finite-precision
    premise. Set $q\teq a_0\varphi$. If $\mu_{1,i_1}=0$ or some
    $\mu_{k,i_k}=0$, the evaluator returns acceptance zero exactly. Otherwise
    the ideal log-acceptance ratio is
    \[
        g_{\idx}\teq
        \log W_{\idx}
        +\sum_{k=1}^m\min\{0,\log\mu_{k,i_k}-u_{k,i_k}\}
        -u_Z-\log\mu_{1,i_1}=\log A_{\idx}\le0.
    \]
    Approximate each of the at most $m+1$ logarithms of positive represented
    marginal masses to error at most $q/(64(m+1))$, and form the remaining
    fixed-point sum exactly. Since $x\mapsto\min\{0,x\}$ and projection onto
    $(-\infty,0]$ are nonexpansive, the resulting exponent
    $\widehat g_{\idx}\le0$ satisfies
    $|\widehat g_{\idx}-g_{\idx}|\le q/64$.

    Compute a downward-rounded threshold
    $t_q\le\log(q/32)$ with $t_q\ge\log(q/64)$. If
    $\widehat g_{\idx}<t_q$, return acceptance zero; then the ideal acceptance
    probability is at most $q/32+q/64=3q/64$. Otherwise evaluate
    $e^{\widehat g_{\idx}}$ to absolute error at most $q/64$. The exponential
    is $1$-Lipschitz on $(-\infty,0]$, so in either branch the computed
    $\widetilde A_{\idx}\in[0,1]$ satisfies
    \[
        |\widetilde A_{\idx}-A_{\idx}|
        \le\frac{3q}{64}<\frac{a_0\varphi}{8}.
    \]
    Synthesize the controlled rotation to the operator-norm tolerance required
    by Lemma~\ref{lem:finite_precision_rejection}; its hybrid error is charged
    separately there. All comparisons with zero use the exact marginal
    encodings, so no evaluation of $\log0$ occurs.
    Each accepted sample then uses
    $\widetilde O(n^{(m-1)/2}\log(N_{\rm samp}/\zeta))$ weight queries.
    A union bound over amplification failures and the tensorization bound
    for total variation charge at most $\zeta/4$ to the full sampling
    stage. Since
    $\log(q^{-1})=O(m\log n+\log(N_{\rm samp}/\zeta))$, the directed
    logarithms, threshold, and exponential require
    \[
        O\!\left(m\log n+\log\frac{N_{\rm samp}}\zeta\right)
    \]
    additional fractional bits beyond the represented inputs.

    Draw
    \[
        N_{\rm samp}=\left\lceil c_{\rm samp}m^2n\theta^{-2}\log\frac{m}{\zeta}\right\rceil
    \]
    independent samples $I_1,\ldots,I_{N_{\rm samp}}$ and set $\widehat R=N_{\rm samp}^{-1}\sum_{t=1}^{N_{\rm samp}}e_{I_t}$.  For $p_k\teq\proj_k(R)$, the standard empirical-distribution estimate gives
    \[
        \mathbb E\|\proj_k(\widehat R)-p_k\|_1
        \le\sqrt{\frac n{N_{\rm samp}}}.
    \]
    Changing one sample changes $\sum_k\|\proj_k(\widehat R)-p_k\|_1$ by at most $2m/N_{\rm samp}$.  McDiarmid's inequality and the choice of $N_{\rm samp}$ therefore imply
    \begin{equation}\label{e:empirical_marginal_bound}
        \sum_{k=1}^m
        \|\proj_k(\widehat R)-\proj_k(R)\|_1
        \le\theta
    \end{equation}
    except with probability at most $\zeta/4$, for a sufficiently large absolute constant $c_{\rm samp}$.  For the fixed tensor $C$, the same sample size and Hoeffding's inequality give
    \begin{equation}\label{e:empirical_cost_bound}
        \langle C,\widehat R\rangle
        \le\mathbb E_{I\sim R}[C_I]+\theta\|C\|_\infty
    \end{equation}
    except with probability at most $\zeta/4$.

    Set $V^{(0)}\teq s_0\widehat R$. For $k=1,\ldots,m$ and each $a$, compute the excess
    \[
        r_{k,a}\teq
        \bigl(\proj_k(V^{(k-1)})_a-\mu_{k,a}\bigr)_+.
    \]
    Traverse the listed atoms with $i_k=a$ in a fixed order and remove mass from them until exactly $r_{k,a}$ has been removed; the last visited atom may be only partially reduced. Denote the resulting tensor by $V^{(k)}$. At step $k$, the $k$th marginal becomes the coordinatewise minimum of $\proj_k(V^{(k-1)})$ and $\mu_k$. Marginals processed earlier remain below their targets because
    \[
        V^{(m)}\le\cdots\le V^{(0)}
    \]
    entrywise.  The mass removed at step $k$ is
    \begin{align*}
        \|V^{(k-1)}\|_1-\|V^{(k)}\|_1
        &=\sum_a
        \bigl(\proj_k(V^{(k-1)})_a-\mu_{k,a}\bigr)_+\\
        &\le\sum_a
        \bigl(s_0\proj_k(\widehat R)_a-\mu_{k,a}\bigr)_+\\
        &\le\sum_a
        \bigl(s_0\proj_k(\widehat R)_a-s_0\proj_k(R)_a\bigr)_+\\
        &=\frac {s_0}2
        \|\proj_k(\widehat R)-\proj_k(R)\|_1,
    \end{align*}
    where the second inequality uses
    $s_0\proj_k(R)=(s_0/s)\proj_k(Q^\downarrow)\le\mu_k$.
    Set $X_0\teq V^{(m)}$.  Summing over $k$ and applying \eqref{e:empirical_marginal_bound} gives
    \begin{equation}\label{e:empirical_retained_mass}
        \proj_k(X_0)\le\mu_k\quad(k\in[m]),
        \qquad
        \|X_0\|_1\ge s_0\left(1-\frac\theta2\right).
    \end{equation}
    Moreover, $X_0\le V^{(0)}$, so \eqref{e:empirical_cost_bound}, $s_0\le s$, and $Q^\downarrow\le Q$ imply
    \begin{equation}\label{e:empirical_subcoupling_cost}
        \langle C,X_0\rangle
        \le \frac{s_0}{s}\langle C,Q^\downarrow\rangle
        +s_0\theta\|C\|_\infty
        \le\langle C,Q\rangle+\theta\|C\|_\infty.
    \end{equation}

    It remains to complete $X_0$ without materializing a tensor product.  Set $d_k=\mu_k-\proj_k(X_0)$.  These vectors are nonnegative and all have mass $1-\|X_0\|_1$. Because the remaining masses stay equal after every subtraction, either all deficits are zero or every $d_k$ has a positive coordinate. Repeatedly choose one such coordinate $a_k$ from every $d_k$, place weight $w=\min_k d_{k,a_k}$ on $(a_1,\ldots,a_m)$, and subtract $w$ from the chosen coordinates.  Each iteration zeros at least one active coordinate, and the final iteration zeros at least $m$ of them.  The resulting completion $D$ consequently has
    \[
        |\operatorname{supp}(D)|\le\sum_{k=1}^m|\operatorname{supp}(d_k)|-m+1\le mn-m+1,
        \qquad \proj_k(D)=d_k.
    \]
    Thus $\widetilde X=X_0+D$ is exactly feasible.  Its completion cost is at most
    \[
        (1-\|X_0\|_1)\|C\|_\infty
        \le\left(\frac58\delta+\frac\theta2\right)\|C\|_\infty.
    \]
    Combining this bound with \eqref{e:empirical_subcoupling_cost} proves the claimed cost bound. The empirical tensor has at most $N_{\rm samp}$ atoms, proving the support bound. Multiplying the cost of one amplified sample by $N_{\rm samp}$ proves the sampling-query bound; forming the empirical marginals, applying the $m$ trimming passes, and running the greedy completion take the stated classical time.

    Finally, suppose every represented marginal mass is a multiple of $2^{-b_\mu}$. Choose the constant-factor tolerances $\delta$ and $\theta$ to be dyadic numbers within a factor two of their prescribed values, and let $b_s$ be the fractional bit length of the dyadic number $s_0=1-5\delta/8$. All entries of $V^{(0)}$ and all target marginal masses then lie on the common grid
    \[
        \Gamma^{-1}\mathbb Z,
        \qquad
        \Gamma=N_{\rm samp}2^{\max\{b_\mu,b_s\}}.
    \]
    The additive trimming pass removes minima of grid masses, and the greedy completion does the same, so every output weight remains on this grid. Each weight uses
    \[
        O\!\left(\max\{b_\mu,b_s\}
        +\log(mn\theta^{-1}\zeta^{-1})\right)
    \]
    bits. For the choices in Theorem~\ref{t:quantum_recovery_full},
    $b_s=O(b_C+b_\varepsilon+\log(m\kappa))$, which gives its stated
    output-length bound. This also shows that exact feasibility does not rely
    on a unit-cost operation on unbounded rationals.
\end{proof}

\begin{theorem}[Quantum sparse coupling recovery]
\label{t:quantum_recovery_full}
Let $(C,\musf,\varepsilon)$ be a represented instance as in
Section~\ref{s:finite_input_appendix}, and define
$\kappa=\max\{1,(\max C-\min C)/\varepsilon\}$. Under coherent cost access
\eqref{e:cost_oracle_main}, coherent lookup of the length-$mn$ working
vectors and represented marginal masses, and target-marginal state
preparation, a quantum algorithm returns, with probability at least $2/3$, a
classical list representing $\widetilde X$ such that
\[
    \proj_k(\widetilde X)=\mu_k\quad(k\in[m]),
    \qquad
    \langle C,\widetilde X\rangle
    \le\mathsf{MOT}(C,\musf)+\varepsilon.
\]
The list has $O(m^2n\kappa^2\log(mn\kappa)+mn)$ atoms, and the algorithm
uses $\widetilde O(m^4n^{(m+1)/2}\kappa^3)$ coherent cost queries and
$\widetilde O(m^5n^{(m+1)/2}\kappa^3)$ working-vector lookups and unit-cost
arithmetic operations. With $b_{\rm sink}$ as in
Theorem~\ref{t:quantum_sinkhorn_full}, its non-oracle gate complexity under
quadratic reversible fixed-point arithmetic is
\[
    \widetilde O\!\left(
        m^5n^{\frac{m+1}{2}}\kappa^3b_{\rm sink}^2
    \right),
\]
of which the recovery stage contributes only
$\widetilde O(m^3n^{(m+1)/2}\kappa^2b_{\rm sink}^2)$. This is
$\widetilde O(m^7n^{(m+1)/2}\kappa^3)$ when
$b_C+b_\mu+b_\varepsilon=O(m\log(n\kappa))$. Compilation adds the
access costs in Remark~\ref{rem:no_qram}. The output has
\[
    b_{\rm out}=O\!\left(
        (m^2n\kappa^2\log(mn\kappa)+mn)
        (m\log n+b_C+b_\mu+b_\varepsilon+\log(mn\kappa))
    \right)
\]
bits and takes $O(b_{\rm out})$ classical bit operations to write.
\end{theorem}

\begin{algorithm}[t]
\caption{Quantum Sinkhorn with sparse exact recovery}
\label{alg:quantum_recovery_main}
\begin{algorithmic}[1]
\Require Marginals $\musf$, coherent cost and marginal-state access, accuracy $\varepsilon$
\State Compute $C_{\min}\gets\min C$ and $C_{\max}\gets\max C$ by quantum min/max finding~\cite{durr1996quantum}
\State Replace $C\gets C-C_{\min}\mathbf1$ and set $\Delta_C\gets C_{\max}-C_{\min}$
\If{$\Delta_C\le\varepsilon$}
    \State \Return $\textsc{GreedyCouple}(\mu_1,\ldots,\mu_m)$
\EndIf
\State Choose dyadic $\eta,\delta,\theta$ as in Theorem~\ref{t:quantum_recovery_full}
\State Form full-support marginals $\musf'$ with total perturbation at most $\delta/4$
\State $\beta\gets\textsc{QuantumSinkhorn}(K=e^{-C/\eta},\musf',\delta/4)$
\If{$\beta=\textsc{Failure}$}
    \State \Return $\textsc{GreedyCouple}(\mu_1,\ldots,\mu_m)$
\EndIf
\State Estimate upper marginals $\bar\nu_{k,a}$ of $Q_\beta$ and set $f_{k,a}\gets\min\{1,\mu_{k,a}/\bar\nu_{k,a}\}$
\State Represent $Q^\downarrow_{\idx}\gets Q_\beta(\idx)\prod_k f_{k,i_k}$ implicitly and set $s_0\gets1-5\delta/8$
\State Set $N_{\rm samp}\gets\widetilde\Theta(m^2n\theta^{-2})$ and $P_{\idx}\gets\mu_{1,i_1}/n^{m-1}$
\State Draw $I_1,\ldots,I_{N_{\rm samp}}\sim Q^\downarrow/\|Q^\downarrow\|_1$ by amplitude amplification from $P$
\State $X_0\gets\textsc{Trim}\bigl((s_0/N_{\rm samp})\sum_{t=1}^{N_{\rm samp}} e_{I_t},\musf\bigr)$ so $\proj_k(X_0)\le\mu_k$ for all $k$
\State $d_k\gets\mu_k-\proj_k(X_0)$ and $D\gets\textsc{GreedyCouple}(d_1,\ldots,d_m)$
\State \Return the sparse list representing $X_0+D$
\end{algorithmic}
\end{algorithm}

\begin{proof}[Proof of Theorem~\ref{t:quantum_recovery_full}]
    Let $\Delta_C=\max C-\min C$ and $\kappa=\max\{1,\Delta_C/\varepsilon\}$. We combine Theorem~\ref{t:quantum_sinkhorn_full} with Lemma~\ref{lem:sparse_gibbs_recovery_full}.
    Compute $L=\min C$ and $U=\max C$ by quantum minimum and maximum finding~\cite{durr1996quantum,van2020quantum}, jointly boosted to failure probability at most $1/9$, so that $\Delta_C=U-L$, and write $C^+=C-L\mathbf1$. This preprocessing uses $\widetilde O(\sqrt N)$ cost queries. If $\Delta_C\le\varepsilon$, greedily couple the marginals as in the last step of Lemma~\ref{lem:sparse_gibbs_recovery_full}. This gives an exactly feasible coupling with at most $mn-m+1$ atoms whose cost exceeds the optimum by at most $U-L\le\varepsilon$. Otherwise set
    \[
        \eta\teq\operatorname{dyad}\!\left(\frac{\varepsilon}{4mq_n}\right),
        \qquad
        \delta=\theta\teq
        \operatorname{dyad}\!\left(\frac{\varepsilon}{8\Delta_C}\right).
    \]
    Set $\lambda\teq\operatorname{dyad}(\delta/(8m))$ and
    $\mu'_k=(1-\lambda)\mu_k+\lambda u_n^{\rm dyad}$. Then
    \[
        \sum_{k=1}^m\|\mu'_k-\mu_k\|_1\le2m\lambda\le\delta/4,
        \qquad
        \min_{k,a}\mu'_{k,a}\ge\frac{\lambda}{2n}>\frac{\delta}{32mn}.
    \]
    Run the algorithm of Corollary~\ref{c:quantum_entropic_mot_tensor} for $C^+$ with residual $\delta/4$ against $\musf'$; it fails with probability at most $1/9$. Otherwise it returns potentials defining a normalized Gibbs tensor $Q$ such that, by the triangle inequality,
    \[
        E_{\musf}(Q)
        \le
        \sum_{k=1}^m\|\proj_k(Q)-\mu'_k\|_1
        +\sum_{k=1}^m\|\mu'_k-\mu_k\|_1
        \le\delta/2.
    \]
    Lemma~\ref{lem:gibbs_pre_recovery}, applied to $C^+$, gives
    \[
        \langle C^+,Q\rangle
        \le \textup{\textsf{MOT}}(C^+,\musf)
        +2\Delta_C E_{\musf}(Q)+\eta m\log n.
    \]
    Apply Lemma~\ref{lem:sparse_gibbs_recovery_full} to $C^+$ with $\zeta=1/9$ and the choices above. A union bound over range finding, Sinkhorn, and recovery gives overall success probability at least $2/3$. Since $\delta,\theta\le\varepsilon/(8\Delta_C)$, the entropic, pre-recovery marginal, and recovery contributions are at most $\varepsilon/4$, $\varepsilon/8$, and $17\varepsilon/64$, so the total is at most $41\varepsilon/64<\varepsilon$. Since every feasible coupling has mass one, shifting back adds $L$ to both costs and preserves this guarantee for $C$.

    The potential computation has the query and gate costs of Theorem~\ref{t:quantum_sinkhorn_full}. The final one-sided marginal and log-normalizer estimates use the same primitives and are lower order. Lemma~\ref{lem:sparse_gibbs_recovery_full} adds $\widetilde O(m^2n^{(m+1)/2}\kappa^2)$ weight queries and outputs $O(m^2n\kappa^2\log(mn\kappa)+mn)$ atoms. For constant recovery failure probability, Lemma~\ref{lem:finite_precision_rejection} evaluates each acceptance probability through its logarithm: $b_C+b_\varepsilon$ covers the exact shifted cost and dyadic parameters, $b_\mu$ covers marginal lookup and state preparation, the log-acceptance exponent has magnitude $\widetilde O(m^2\kappa\log n)$ and therefore needs only logarithmically many integer bits, and $O(m\log n+\log(N_{\rm samp}/\zeta))$ fractional bits cover exponentiation and sampling error. These word lengths are $O(b_{\rm sink})$ up to logarithmic factors.
    The reversible arithmetic of recovery therefore adds
    $\widetilde O(m^3n^{(m+1)/2}\kappa^2b_{\rm sink}^2)$ gates. The
    $\widetilde O(\sqrt N)$ range computation, $\widetilde O(n)$
    marginal-state-preparation preprocessing, and recovery query costs are
    dominated by the displayed bounds. Exact feasibility depends only on
    common-grid additive trimming and greedy deficit completion;
    finite-precision state preparation affects the sampling distribution
    and failure budget, not the marginal identities of the returned list.
    The final paragraph of the proof of Lemma~\ref{lem:sparse_gibbs_recovery_full} gives the
    stated weight and output bit lengths.
\end{proof}

Algorithm~\ref{alg:quantum_recovery_main} summarizes the pipeline. For fixed $m\ge3$ and constant $\kappa$, the input has $N=n^m$ entries, whereas the algorithm uses $\widetilde O(n^{(m+1)/2})$ queries and returns $\widetilde O(n)$ atoms: it needs $\widetilde O(n)$ samples, each costing $\widetilde O(n^{(m-1)/2})$ queries. More generally, the bound on the list size is $o(N)$ exactly when $m^2n\kappa^2\log(mn\kappa)+mn=o(n^m)$. Because the list is exactly feasible and its cost can be evaluated with one cost query per atom, independent repetitions boost the success probability, and solving the MOT program restricted to its support leaves at most $mn-m+1$ atoms without increasing the cost.

\begin{proposition}[Boosting and support purification]\label{p:boost_purify}
Let $s$ denote the list size in Theorem~\ref{t:quantum_recovery_full}.
\begin{enumerate}[label=(\roman*),leftmargin=*]
\item For every $\rho\in(0,1/3)$, running the algorithm of
Theorem~\ref{t:quantum_recovery_full} $O(\log(1/\rho))$ times and returning
the cheapest output gives an exactly feasible additive-$\varepsilon$ coupling
with probability at least $1-\rho$, using $O(s\log(1/\rho))$ additional
cost queries.
\item Let $\widetilde X$ be an exactly feasible list with support $S$. Given
the costs $C_{\idx}$ for $\idx\in S$, a basic optimal solution $X'$ of the
MOT program restricted to $S$ satisfies $\proj_k(X')=\mu_k$ for every $k$,
$\operatorname{supp}(X')\subseteq S$, $|\operatorname{supp}(X')|\le mn-m+1$,
and $\langle C,X'\rangle\le\langle C,\widetilde X\rangle$. It can be
computed exactly with no further cost queries in time polynomial in $|S|$,
$mn$, and the encoding lengths of the listed costs and masses.
\end{enumerate}
\end{proposition}
\begin{proof}
For (i), every run returns an exactly feasible list: trimming and deficit
completion enforce the marginal identities on every execution path, and a run
that aborts can output the greedy coupling of the marginals instead. Querying
the cost of every listed atom evaluates each output cost exactly. The
cheapest output is additive-$\varepsilon$ optimal whenever at least one run
succeeds, which happens except with probability $(1/3)^r\le\rho$ after
$r=\lceil\log_3(1/\rho)\rceil$ runs.

For (ii), the restricted program
\[
    \min\Bigl\{\langle C,X\rangle:\ X\ge0,\ \operatorname{supp}(X)\subseteq S,\ \proj_k(X)=\mu_k\ (k\in[m])\Bigr\}
\]
is a linear program in $|S|$ variables. Its feasible set is bounded and
contains $\widetilde X$, so it has an optimal vertex $X'$ with
$\langle C,X'\rangle\le\langle C,\widetilde X\rangle$. The positive entries
of a vertex have linearly independent constraint columns, and
Lemma~\ref{lem:sparse_basic_mot} bounds the rank of these columns by
$mn-m+1$. Only the listed costs enter the program, and an exact
polynomial-time linear-programming algorithm that returns a basic optimal
solution~\cite{grotschel1988geometric} computes $X'$ in the stated time.
\end{proof}

\begin{remark}[Compilation without qRAM]\label{rem:no_qram}
Let $T_C,T_W,T_\mu,T_U$ denote the gate costs of one cost query,
working-vector lookup, marginal-state preparation or mass lookup, and
working-vector write, respectively. If the recovery algorithm makes $q_C$
cost queries and $q_W$ working-vector lookups, compiling these interfaces
and maintaining its working data adds
\[
    \widetilde O\!\left(
        q_CT_C+q_W(T_W+T_\mu)+m^3n\kappa^2T_U+nb_\mu
    \right)
\]
gates, while writing the output list takes $O(b_{\rm out})$ classical bit
operations. The query bounds do not depend on how classical data are stored.
Suppose $C_{\idx}=\mathbf c(z_{1,i_1},\ldots,z_{m,i_m})$ and the supports,
potentials, and represented marginal masses are stored as classical tables.
A multiplexed table-lookup circuit reads one word of an $n$-word table with
$\widetilde O(n)$ gates~\cite{babbush2018encoding}, and a potential update
only changes the classical description of that circuit. Hence
$T_W=\widetilde O(n)$ and $T_C=\widetilde O(mn)+T_{\mathbf c}$, where
$T_{\mathbf c}$ is the gate cost of one reversible evaluation of $\mathbf c$;
writes cost $O(n)$ classical work per update instead of $T_U$; and a
multiplexed rotation tree prepares $\sum_a\sqrt{\mu_{1,a}}\ket a$ and looks
up masses with $T_\mu=\widetilde O(n)$. With
$q_C=\widetilde O(m^4n^{(m+1)/2}\kappa^3)$ and
$q_W=\widetilde O(m^5n^{(m+1)/2}\kappa^3)$ from
Theorem~\ref{t:quantum_recovery_full}, compilation therefore adds
$\widetilde O(m^5n^{(m+3)/2}\kappa^3+m^4n^{(m+1)/2}\kappa^3T_{\mathbf c})$
gates. For fixed $m$,
$\kappa\le\operatorname{polylog}(n)$, $b_\mu\le\operatorname{polylog}(n)$,
and $T_{\mathbf c}=\operatorname{polylog}(n)$, the total gate count is
$\widetilde O(n^{(m+3)/2})$, which is sublinear in $N=n^m$ for every
$m\ge4$ without quantum random-access memory. For $m\le3$ this naive
compilation is not sublinear, and sublinear gate bounds require lookups with
polylogarithmic cost.
\end{remark}

\section{Comparison with existing quantum approaches}\label{s:comparison}
Our quantum Sinkhorn algorithm is closest to quantum algorithms for matrix scaling and matrix balancing~\cite{vanapeldoorn2021scaling,gribling2022scaling,nieuwboer2024classical}. In the bimarginal case, entropic OT is a matrix-scaling problem: one scales a positive matrix so that its two marginals match prescribed vectors. For MOT with $m\ge3$, one scales a positive order-$m$ tensor to match $m$ tensor marginals, and each update estimates fiber sums over $n^{m-1}$ entries instead of row or column sums over $n$ entries. We reuse the quantum log-sum-exp primitive of the matrix-scaling work and combine it with the multimarginal Sinkhorn geometry and rounding analysis of Lin et al.~\cite{lin2022complexity} and with the noisy-descent and drift analysis of Section~\ref{s:quantum_tensor_sinkhorn}.  

\begin{table}[t]
\caption{Quantum approaches to multimarginal OT discussed in this paper. Bounds are stated under each method's own access model and are therefore not all end-to-end comparable.}
\label{tab:quantum_method_comparison}
\centering
\small
\setlength{\tabcolsep}{3pt}
\renewcommand{\arraystretch}{1.08}
\begin{tabularx}{\textwidth}{@{}
    >{\raggedright\arraybackslash}p{0.24\textwidth}
    >{\raggedright\arraybackslash}p{0.17\textwidth}
    >{\raggedright\arraybackslash}p{0.2\textwidth}
    >{\raggedleft\arraybackslash}X@{}}
\toprule
Method and scope & Reference & Output & Stated resource bound \\
\midrule
Tall-LP IPM (any $m$)
& \cite{apers2023quantum}
& Generic LP solution
& $\widetilde O(n^{m/2}(mn)^5)$ row queries \\
\midrule
\textbf{Tensor Sinkhorn (ours; any $m$)}
& Theorem~\ref{t:quantum_sinkhorn_full}
& Potentials
& $\boldsymbol{\widetilde O(m^4n^{(m+1)/2}\kappa^3)}$ cost queries \\
\addlinespace[2pt]
\textbf{Sparse recovery (ours; any $m$)}
& Theorem~\ref{t:quantum_recovery_full}
& Exact classical list
& $\boldsymbol{\widetilde O(m^4n^{(m+1)/2}\kappa^3)}$ cost queries \\
\addlinespace[2pt]
\textbf{Projected subgradient (ours; any $m$)}
& Theorem~\ref{t:quantum_md}
& Value and exact dual certificate
& $\boldsymbol{\widetilde O(m^2n^{m/2+1}\kappa^2)}$ cost queries \\
\bottomrule
\end{tabularx}
\end{table}  

For multimarginal OT, one could also apply the tall-LP quantum interior-point method of Apers and Gribling~\cite{apers2023quantum} to the dual formulation  \eqref{e:MOT_dual_tensor}. A row query to this LP is sparse: the constraint indexed by $\idx$ contains exactly $m$ nonzero coefficients and one query to $C_{\idx}$. Thus the Apers--Gribling Lewis-weight IPM applies with $M$ variables and $N+2M$ inequalities, giving row-query complexity
\[
    \widetilde{O}(\sqrt{N+2M}\,M^5)
    = \widetilde{O}\!\left(n^{\frac{m}{2}}(mn)^5\right) = \widetilde{O}\!\left(n^{\frac{m+10}{2}} m^5\right),
\]
and time $\widetilde{O}(\sqrt{N+2M}\,\poly(M))$; their volumetric-barrier bound gives $\widetilde{O}((N+2M)^{\frac{3}{4}}M^{\frac{5}{4}})$ row queries. These IPM bounds have only polylogarithmic dependence on $1/\varepsilon$, but the generic method pays a large polynomial overhead in the smaller dimension $M=mn$. Therefore, the off-the-shelf quantum IPM is not competitive in the low-precision regime targeted here; improving this comparison would require an MOT-specific IPM implementation that substantially reduces the $\poly(M)$ overhead.

Classical methods can also avoid dense iterates when the cost has additional structure. For general $m$, the structured-cost algorithms of Altschuler and Boix-Adser{\`a}~\cite{altschuler2023structured} exploit succinct dual-feasibility oracles. In the general entry-access model with fixed $m$, by contrast, Theorem~\ref{thm:sparse_recovery_lower} rules out classical algorithms that output an approximately optimal coupling with $o(N/\log^2 n)$ cost queries, whatever their output size. On a structured subclass that admits a classical sampler for the trimmed Gibbs law, that sampler is the relevant classical baseline.

\section{Lower bounds}\label{s:lower_bounds}
We now prove lower bounds in the unstructured cost-query model. Section~\ref{s:sparse_lower} treats constructive sparse recovery and shows that the upper bounds of Sections~\ref{s:mart}--\ref{s:sparse_recovery} have optimal exponents for fixed $m$. Section~\ref{s:value_lower} treats value estimation.

\subsection{Sparse coupling recovery}\label{s:sparse_lower}
For sparse outputs the serialization lower bound does not apply, so we instead force any successful algorithm to reveal many low-cost entries. Draw independent marks with probability $p=\log^2(n)/n^{m-1}$ and set $C_{\idx}=1$ on unmarked entries and $C_{\idx}=0$ on marked entries. For fixed $m$, the marked $m$-partite hypergraph has degree $(1+o(1))\log^2 n$ and codegree $o(\log^2 n)$ with high probability. Pippenger's theorem~\cite[Theorem~1.1]{ehard2020pseudorandom} then supplies a marked matching of size at least $(1-\varepsilon_0)n$, so the MOT value is at most $\varepsilon_0$.

Every atom in a coupling with uniform marginals has mass at most $1/n$; if the marginals are only approximately uniform, the atoms sharing a first coordinate carry at most $1/n$ plus that coordinate's marginal error. A low-cost output must therefore list $\Theta(n)$ distinct marked atoms. Classically, no verification is needed: unqueried entries are marked only with probability $p$, and $q$ inspected locations contain only $pq$ marks in expectation, which gives $q=\Omega(n/p)=\Omega(n^m/\log^2n)$. Quantumly, querying the reported support verifies those marks and reduces recovery to multi-solution Bernoulli search, whose lower bound~\cite[Theorem~3.1]{cojocaru2023multisolution} gives $q=\Omega(n/\sqrt p)=\Omega(n^{(m+1)/2}/\log n)$.

\begin{theorem}[Sparse coupling-recovery lower bound]\label{thm:sparse_recovery_lower}
Fix a constant $m\ge2$, a constant $c\ge0$, and a constant normalized accuracy $\varepsilon_0\in(0,1/10)$. Consider algorithms that, for every Boolean cost tensor $C\in\{0,1\}^{[n]^m}$ with uniform marginals $\musf$, return with probability at least $2/3$ a list of at most $n(\log n)^c$ support atoms and masses describing a tensor $X\ge0$ with $\langle C,X\rangle\le\mathsf{MOT}(C,\musf)+\varepsilon_0$ and $\sum_k\|\proj_k(X)-\mu_k\|_1\le\varepsilon_0$; exactly feasible outputs are the special case of zero marginal error. For all sufficiently large $n$, every such randomized classical algorithm requires $\Omega(n^m/\log^2 n)$ cost-entry queries, even without the bound on its output size, and every such quantum algorithm requires $\Omega(n^{(m+1)/2}/\log n)$ cost-entry queries. The constants and the threshold for $n$ may depend on $m,c,$ and $\varepsilon_0$.
\end{theorem}
On the nonconstant Boolean hard instances, $\Delta_C=1$, so additive and normalized errors coincide.

\begin{proof}
    Draw independent bits $f_{\idx}\sim\operatorname{Bernoulli}(p)$ for
    \[
        p\teq\frac{\log^2 n}{n^{m-1}},
        \qquad C_{\idx}\teq1-f_{\idx}.
    \]
    Regard the marked entries as the edges of an $m$-partite, $m$-uniform hypergraph $\mathcal H$ with $n$ vertices in every part.  A vertex degree is binomial with mean $d=\log^2 n$.  Chernoff's inequality and a union bound over the $mn$ vertices give
    \[
        \min\deg(\mathcal H)=(1-o(1))d,
        \qquad \max\deg(\mathcal H)=(1+o(1))d
    \]
    with probability $1-o(1)$.  Two vertices in the same part have codegree zero.  Two vertices in different parts have binomial codegree with mean $d/n$, and the standard binomial tail bound gives
    \[
        \Pr[\operatorname{codeg}_{\mathcal H}(u,v)\ge\log n]
        \le\left(\frac{e\log n}{n}\right)^{\log n}.
    \]
    A union bound over all vertex pairs therefore gives $\max\operatorname{codeg}(\mathcal H)\le\log n=o(d)$ with probability $1-o(1)$.

    Apply Pippenger's almost-perfect matching theorem~\cite[Theorem~1.1]{ehard2020pseudorandom} with uncovered fraction $\varepsilon_0$.  For fixed $m$ and $\varepsilon_0$, its regularity and codegree constants are fixed, while the preceding estimates give $\min\deg(\mathcal H)/\max\deg(\mathcal H)=1-o(1)$ and $\max\operatorname{codeg}(\mathcal H)/\max\deg(\mathcal H)=o(1)$.  Thus its hypotheses hold for all sufficiently large $n$, and with probability $1-o(1)$ the hypergraph $\mathcal H$ has a matching of at least $(1-\varepsilon_0)n$ marked edges.  Place mass $1/n$ on each edge in this matching and couple the unmatched vertices arbitrarily.  This is feasible for the uniform marginals and has cost at most $\varepsilon_0$, so
    \begin{equation}\label{e:bernoulli_low_optimum}
        \textup{\textsf{MOT}}(C,u_n^{\otimes m})\le\varepsilon_0
    \end{equation}
    on the same event.

    Run a putative algorithm on $C=1-f$, and suppose that \eqref{e:bernoulli_low_optimum} and the algorithm's guarantee both hold. Let $\mathcal M$ be the set of marked entries and $\theta\teq\sum_k\|\proj_k(X)-\mu_k\|_1\le\varepsilon_0$. Then $1-\theta\le\|X\|_1=\|\proj_1(X)\|_1\le1+\theta$, and $\langle C,X\rangle$ is the mass of $X$ on unmarked entries, so
    \[
        X(\mathcal M)\ge\|X\|_1-2\varepsilon_0\ge1-3\varepsilon_0.
    \]
    Any set $S$ of entries has at most $|S|$ distinct first coordinates, so
    \begin{equation}\label{e:approx_atom_cap}
        X(S)\le\frac{|S|}{n}+\theta.
    \end{equation}
    Applying \eqref{e:approx_atom_cap} to $S=\mathcal M\cap\operatorname{supp}(X)$ shows that the output lists at least $(1-4\varepsilon_0)n$ distinct marked entries. Exact feasibility is the special case $\theta=0$.

    For the quantum claim, suppose the algorithm uses $q_{\rm rec}$ queries and returns at most $s\teq n(\log n)^c$ atoms. Query $f$ on each listed atom and return any $k\teq\lfloor(1-5\varepsilon_0)n\rfloor$ marked ones. This verification uses at most $s$ additional queries. A quantum algorithm may make these final queries in the computational basis, so the composed procedure is a valid quantum query algorithm using $q_{\rm rec}+s$ queries. Since the guarantee holds on every input, it solves $k$-solution Bernoulli search with success probability at least $2/3-o(1)$ over the random choice of $f$.

    The quantum multi-solution Bernoulli-search bound of Cojocaru et al.~\cite[Theorem~3.1]{cojocaru2023multisolution} states that an algorithm making $q$ queries has success probability at most $(1/(\pi k))(e^2(q/k)^2p)^k$. Although that result indexes its finite domain by bit strings, it applies here by embedding $[n]^m$ into a power-of-two domain and restricting every query and reported solution to the embedded subset; the Bernoulli law on that subset is unchanged. Constant success thus requires $q=\Omega(k/\sqrt p)$. Applying this with $q=q_{\rm rec}+s$ gives $q_{\rm rec}+s=\Omega(n^{(m+1)/2}/\log n)$.
    For fixed $m\ge2$ and fixed $c$, the allowed support size $s$ is lower order, proving the quantum claim.

    For the classical claim, neither an output bound nor verification is needed. Let the algorithm make $q_{\rm rec}$ adaptive queries, let $\mathcal Q$ be the set of queried locations, and let $\widehat X\teq X\,\mathbbm 1[\|X\|_1\le2]$; both are functions of the algorithm's randomness and the answers it has seen. Conditional on this view, the unqueried bits remain independent $\operatorname{Bernoulli}(p)$ variables, so
    \[
        \mathbb E\bigl[\widehat X(\mathcal M\setminus\mathcal Q)\bigr]
        =p\,\mathbb E\bigl[\widehat X([n]^m\setminus\mathcal Q)\bigr]
        \le2p=o(1).
    \]
    By Markov's inequality, $\widehat X(\mathcal M\setminus\mathcal Q)<\varepsilon_0$ except with probability $o(1)$. On the event above, $\widehat X=X$ because $\|X\|_1\le1+\theta$, so $X(\mathcal M\cap\mathcal Q)>1-4\varepsilon_0$, and \eqref{e:approx_atom_cap} shows that the algorithm has queried at least $k$ marked locations. Every newly queried location is marked with conditional probability $p$, so the expected number of marked locations among the queried ones is at most $pq_{\rm rec}$. By Markov's inequality, finding $k$ of them with probability $2/3-o(1)$ requires $q_{\rm rec}=\Omega(k/p)=\Omega(n^m/\log^2 n)$, proving the classical claim.
\end{proof}

This lower bound matches both constructive upper exponents for every fixed $m$.

\begin{corollary}[Tight sparse-output query complexity]\label{c:constructive_query_full}
Fix $d>0$, $c>0$, a constant $m\ge2$, and $\varepsilon_0\in(0,1/10)$. If $\Delta_C\le1$ and $(\log n)^{-d}\le\varepsilon\le\varepsilon_0$, then constructing an exactly feasible additive-$\varepsilon$ coupling with $O(n(\log n)^c)$ atoms has query complexity $\widetilde\Theta(n^m)$ classically and $\widetilde\Theta(n^{(m+1)/2})$ quantumly, under the access models above.
\end{corollary}
\begin{proof}
The classical upper bound queries all $n^m$ costs and extracts a basic optimum. The quantum upper bound is Theorem~\ref{t:quantum_recovery_full} followed by the purification of Proposition~\ref{p:boost_purify}, which uses $\widetilde O(n)$ additional cost queries and leaves at most $mn-m+1\le n(\log n)^c$ atoms for all sufficiently large $n$. The lower bounds follow from Theorem~\ref{thm:sparse_recovery_lower}, since additive-$\varepsilon$ accuracy implies additive-$\varepsilon_0$ accuracy. Thus both bounds count queries to the same cost-entry oracle; exact Gibbs-exponent computation affects the gate count but not the query exponent.
\end{proof}

Theorem~\ref{t:quantum_recovery_full} applies to arbitrary marginal weights, while Theorem~\ref{thm:sparse_recovery_lower} proves hardness already for uniform marginals. The lower-bound proof does not extend uniformly to growing $m$, because the constants in the almost-perfect matching theorem depend on the hypergraph uniformity. Establishing the same lower bounds for $m\le\operatorname{polylog}(n)$ would require either quantitative growing-uniformity matching estimates or a different hard distribution.

On the hard instances, a low-cost output must expose $\Theta(n)$ marked entries, and marks have density $n^{-(m-1)}$ up to polylogarithmic factors; the matched exponents are therefore $n\cdot n^{m-1}=n^m$ classically and $n\cdot\sqrt{n^{m-1}}=n^{(m+1)/2}$ quantumly. The quantum exponent is the standard multi-solution Grover exponent, and our algorithm attains it constructively; it gives no speedup beyond Grover scaling. Value estimation does not require exposing $\Theta(n)$ marks and can be hard even when the instance hides only one low-cost block (Section~\ref{s:value_lower}).

\begin{remark}[Classical analogue of the pipeline]\label{rem:classical_analogue}
Every step of the pipeline has a classical counterpart. An exact classical
Sinkhorn update costs $\Theta(n^m)$ kernel evaluations, and ordinary
rejection sampling from the proposal $P$ of
Lemma~\ref{lem:sparse_gibbs_recovery_full} accepts with probability
$\Theta(n^{-(m-1)})$, so each accepted sample costs $\Theta(n^{m-1})$ weight
evaluations in expectation. The classical version of
Theorem~\ref{t:quantum_recovery_full} therefore returns the same kind of
sparse, exactly feasible coupling in $\widetilde O(\poly(m)\,n^m\kappa^2)$
time. This is no better than reading the tensor, as it must be by the
classical part of Theorem~\ref{thm:sparse_recovery_lower}. The quantum
algorithm replaces every fiber sum and every rejection loop by its amplified
counterpart, turning each $n^{m-1}$ search factor into $n^{(m-1)/2}$ while
leaving the factor $n$ from marginal coordinates or output atoms unchanged. For fixed $m$,
Theorem~\ref{thm:sparse_recovery_lower} shows that neither exponent can be
improved under unstructured access, so the gain is exactly the quadratic
search speedup on these steps.
\end{remark}

\subsection{Value estimation}\label{s:value_lower}
The next reduction interpolates between constant and fine accuracy. At constant additive accuracy it gives a search lower bound over $n^{m-1}$ hidden blocks; below the single-entry mass scale $1/n$, it recovers search over all $n^m$ tensor entries.

Identify $[n]$ with the cyclic group $\mathbb{Z}_n$. For each shift vector $s=(s_2,\ldots,s_m)\in\mathbb{Z}_n^{m-1}$, define the perfect $m$-partite matching
\[
    P_s \teq \{(i,i+s_2,\ldots,i+s_m): i\in\mathbb{Z}_n\}\subseteq [n]^m,
\]
with addition modulo $n$. The sets $\{P_s:s\in\mathbb{Z}_n^{m-1}\}$ partition $[n]^m$.

\begin{theorem}[Value-query lower bound for MOT]\label{thm:query_lower_bound_mot}
    Fix $n\geq2$ and $m\geq2$, and let all marginals be uniform: $\mu_1=\cdots=\mu_m=u_n$. For any $\varepsilon\in(0,1/2)$, any bounded-error randomized classical algorithm that, for every Boolean cost tensor $C\in\{0,1\}^{[n]^m}$, outputs an additive $\varepsilon$-approximation to \eqref{e:MOT} requires
    \[
        \Omega\!\left(\frac{n^m}{1+\varepsilon n}\right)
    \]
    cost queries. Any bounded-error quantum algorithm requires
    \[
        \Omega\!\left(\sqrt{\frac{n^m}{1+\varepsilon n}}\right)
    \]
    cost queries.
\end{theorem}
\begin{proof}
    Set $r\teq \lfloor 2\varepsilon n\rfloor+1$, so $1\leq r\leq n$ and $r/n>2\varepsilon$. Split each matching $P_s$ into $N_{\rm seg}\teq\lfloor n/r\rfloor$ disjoint consecutive blocks of size $r$:
    \[
        P_{s,t}\teq \{(i,i+s_2,\ldots,i+s_m): tr\leq i<tr+r\},
        \qquad t=0,\ldots,N_{\rm seg}-1,
    \]
    ignoring any leftover entries. The total number of blocks is $N_{\rm blk}=n^{m-1}N_{\rm seg}=\Omega(n^m/(1+\varepsilon n))$.

    We reduce from promised unstructured search on $N_{\rm blk}$ bits indexed by the pairs $(s,t)$, with promise $|z|\in\{0,1\}$. Given such a search input, define a Boolean cost tensor by setting $C^z_{\idx}=0$ if $\idx\in P_{s,t}$ for the unique indexed block with $z_{s,t}=1$, and $C^z_{\idx}=1$ otherwise. Equivalently, a query to $C^z$ determines the unique matching $P_s$ containing $\idx$, checks whether $\idx$ lies in one of the retained length-$r$ blocks, and if so queries the corresponding bit $z_{s,t}$. Thus one cost query is simulated using at most one query to $z$.

    If $|z|=0$, then $C^z\equiv1$, so every feasible coupling has cost $1$. If $|z|=1$, let $P_{s,t}$ be the marked block. The feasible coupling that places mass $1/n$ on every entry of the full matching $P_s$ has cost $1-r/n$, because exactly the $r$ entries in $P_{s,t}$ have zero cost. Therefore the MOT value is at most $1-r/n$ in the marked case. Since $r/n>2\varepsilon$, an additive $\varepsilon$ value estimate is at least $1-\varepsilon$ in the unmarked case and at most $1-r/n+\varepsilon<1-\varepsilon$ in the marked case. Thus any such algorithm distinguishes the two search promises. The stated lower bounds follow from the randomized and quantum query lower bounds for promised unstructured search on $N_{\rm blk}$ items~\cite{BBBV,zalka1999grover}.
\end{proof}
At constant accuracy the bounds become $\Omega(n^{m-1})$ and $\Omega(n^{(m-1)/2})$, a factor $n$ below the corresponding sparse-recovery bounds.

\section{Discussion and limitations}\label{s:limitations_main}
Both classical routes match the $\Omega(N)$ cost of writing a dense coupling for fixed $\kappa$, up to factors in $m$ and logarithms; packing has the sharper $m$-dependence, and box--simplex is deterministic with $O(m\kappa\log N)$ matrix--vector rounds. For explicit outputs with $n\polylog(n)$ atoms, fixed $m$, and $\kappa\le\polylog(n)$, the query complexity is $\widetilde\Theta(N)$ classically and $\widetilde\Theta(\sqrt{Nn})$ quantumly. The lower bounds survive constant marginal error, and the classical one holds for every output size. The quantum gain comes from multi-solution Grover amplification of fiber sums and rejection loops (Remark~\ref{rem:classical_analogue}).

The quantum query bound carries $m^4\kappa^3$, against $m^2\kappa$ for packing, and its gate count depends on the coherent interfaces: without qRAM, table lookup costs $\widetilde O(n)$ gates per access, which keeps the total sublinear in $N$ only for $m\ge4$ (Remark~\ref{rem:no_qram}). Structured costs change the classical baseline, since graphical, set-optimization, and low-rank costs can admit classical $\poly(n,m)$ algorithms with exact sparse outputs~\cite{altschuler2023structured}; our Boolean hard instances, which require $\Theta(N)$ oracle bits, describe only the unstructured worst case. The lower bounds assume fixed $m$ and constant accuracy, so they do not address growing $m$ or whether the cubic dependence on $\kappa$ is necessary. At constant accuracy, our value-estimation lower bounds are a factor $n$ below the sparse-output bounds, and the complexity of value estimation and of approximate tensor scaling remains open.
 
\section*{Acknowledgements}
 The authors thank Rob Otter and Ruslan Shaydulin for executive support and valuable feedback on this project. We also acknowledge our colleagues at the Global Technology Applied Research Center of JPMorganChase.

\section*{AI use statement}
Generative AI tools were used as research assistants in developing and checking proof arguments for the convergence of noisy quantum multimarginal Sinkhorn, the trim--sample--complete recovery procedure, and the sparse-output lower bound. They were also used to formulate and refine mathematical claims, explore proof strategies, edit portions of the manuscript, and assist with algebraic, finite-precision, failure-probability, and consistency checks. The authors reconstructed and independently verified every AI-assisted argument, checked the cited sources, reviewed all generated text, and take full responsibility for the paper's claims and contents. No datasets, experiments, or empirical results were generated or interpreted using AI.

\section*{Disclaimer}
This paper was prepared for informational purposes by the Global Technology Applied Research center of JPMorgan Chase \& Co. This paper is not a product of the Research Department of JPMorgan Chase \& Co. or its affiliates. Neither JPMorgan Chase \& Co. nor any of its affiliates makes any explicit or implied representation or warranty and none of them accept any liability in connection with this paper, including, without limitation, with respect to the completeness, accuracy, or reliability of the information contained herein and the potential legal, compliance, tax, or accounting effects thereof. This document is not intended as investment research or investment advice, or as a recommendation, offer, or solicitation for the purchase or sale of any security, financial instrument, financial product or service, or to be used in any way for evaluating the merits of participating in any transaction.

\bibliographystyle{alpha}
\bibliography{references}

\appendix
\section{Quantum value estimation by projected subgradient descent}\label{s:quantum}
This appendix gives a quantum algorithm that works with the dual \eqref{e:MOT_dual_tensor} and returns an exactly feasible dual certificate whose objective value is within $\varepsilon$ of $\mathsf{MOT}(C,\musf)$; it does not output a coupling. It uses $\widetilde O(m^2n^{m/2+1}\kappa^2)$ cost queries, compared with $\widetilde O(m^4n^{(m+1)/2}\kappa^3)$ for the constructive algorithm of Theorem~\ref{t:quantum_recovery_full}: its dependence on $\kappa$ is better, while its dependence on $n$ is worse by a factor $\sqrt n$. The subgradients have a closed form, so quantum maximum finding suffices, and the algorithm needs neither general quantum subgradient estimation~\cite{van2020convex,chakrabarti2020quantum} nor quantum subgradient algorithms~\cite{kim2026fast}.

Throughout, we use the finite-input conventions and access model of Section~\ref{s:finite_input_appendix}. By Lemma~\ref{lem:prune_zero_marginals}, we may assume that every marginal has full support; then $n_k\le n$, $M=\sum_kn_k\le mn$, and $N=\prod_kn_k\le n^m$. Pruning preserves the MOT value, so a dual certificate for the pruned instance certifies the same lower bound. It also extends to the original instance by assigning a sufficiently negative common value to the pruned dual coordinates, which have zero mass.

\subsection{Projected subgradient descent with inexact projections}
Let $f:\R{d}\to\R{}$ be convex. A vector $g\in\R{d}$ is a subgradient of $f$ at $x$ if $f(\bar x)\ge f(x)+\langle g,\bar x-x\rangle$ for every $\bar x\in\R{d}$, and $\partial f(x)$ denotes the set of all such subgradients. For a closed convex set $\Xcal\subseteq\R{d}$, let $\Pi_\Xcal$ denote Euclidean projection onto $\Xcal$. We use the standard projected-subgradient bound~\cite[Theorem~3.2]{bubeck2015convex}, stated with an additive perturbation of each projected iterate, which the fixed-point implementation below requires.

\begin{lemma}[Projected subgradient method with inexact projections]\label{t:qpsm_prelim}
Let $\Xcal=\mathcal B_2^d(R)$, let $f:\R{d}\to\R{}$ be convex, and let $x^\star$ minimize $f$ over $\Xcal$. Let $x^{(0)}=0$, $\eta>0$, $e\ge0$, and $T\ge1$. Suppose that, for $t=0,\ldots,T-1$, the vectors $g^{(t)}\in\partial f(x^{(t)})$ satisfy $\|g^{(t)}\|_2\le G$ and the next iterates $x^{(t+1)}\in\Xcal$ satisfy
\[
    \bigl\|x^{(t+1)}-\Pi_\Xcal\bigl(x^{(t)}-\eta g^{(t)}\bigr)\bigr\|_2\le e.
\]
Then $\bar x\teq T^{-1}\sum_{t=0}^{T-1}x^{(t)}$ satisfies
\[
    f(\bar x)-f(x^\star)
    \le\frac{R^2}{2\eta T}+\frac{\eta G^2}{2}+\frac{4Re+e^2}{2\eta}.
\]
In particular, for $e=0$ and $\eta=R/(G\sqrt T)$, the right-hand side is $GR/\sqrt T$.
\end{lemma}
\begin{proof}
Fix $t$ and let $p\teq\Pi_\Xcal(x^{(t)}-\eta g^{(t)})$. Since $x^\star\in\Xcal$ and Euclidean projection onto a closed convex set is nonexpansive,
\[
    \|p-x^\star\|_2^2
    \le\|x^{(t)}-x^\star\|_2^2-2\eta\langle g^{(t)},x^{(t)}-x^\star\rangle+\eta^2G^2.
\]
Both $p$ and $x^\star$ lie in $\Xcal$, so $\|p-x^\star\|_2\le2R$ and
\[
    \|x^{(t+1)}-x^\star\|_2^2
    \le(\|p-x^\star\|_2+e)^2
    \le\|p-x^\star\|_2^2+4Re+e^2.
\]
By convexity, $f(x^{(t)})-f(x^\star)\le\langle g^{(t)},x^{(t)}-x^\star\rangle$. Combining the three displays gives
\[
    f(x^{(t)})-f(x^\star)
    \le\frac{\|x^{(t)}-x^\star\|_2^2-\|x^{(t+1)}-x^\star\|_2^2}{2\eta}
    +\frac{\eta G^2}{2}+\frac{4Re+e^2}{2\eta}.
\]
Summing over $t=0,\ldots,T-1$, using $\|x^{(0)}-x^\star\|_2\le R$, and applying Jensen's inequality to $\bar x$ proves the claim.
\end{proof}

\subsection{The MOT dual as nonsmooth optimization}
Let $A,b,c$ be the MOT data in \eqref{eq:multi-ot-data}. The dual LP \eqref{eq:LP-D} can be recast as the nonsmooth convex problem $\min_{y}\max_{j\in[N]}(A^\top y-c)_j-b^\top y$~\cite{helmberg2000spectral,wolkowicz2012handbook}. For MOT, write $a_{\idx}$ for the column of $A$ indexed by $\idx\in I$, and
\begin{equation}\label{e:value_dual_objective}
    F(y)\teq\max_{\idx\in I}\phi_{\idx}(y)-b^\top y,
    \qquad
    \phi_{\idx}(y)\teq\sum_{k=1}^m y_{k,i_k}-C_{\idx},
    \qquad
    r(y)\teq\max_{\idx\in I}\phi_{\idx}(y).
\end{equation}
Let $\mathbf 1^{(1)}\in\R{M}$ be the indicator of the first marginal block. Every column of $A$ has exactly one nonzero entry in that block, and $\mu_1$ has unit mass, so
\begin{equation}\label{e:first_block_shift}
    A^\top\mathbf 1^{(1)}=\mathbf 1_N,
    \qquad
    b^\top\mathbf 1^{(1)}=1.
\end{equation}
Consequently, $y-r(y)\mathbf 1^{(1)}$ is feasible for \eqref{e:MOT_dual_tensor} and has dual objective value $b^\top y-r(y)=-F(y)$, so minimizing $F$ is equivalent to solving the MOT dual. Shifting only the first block, rather than every block by $r(y)/m$, keeps this operation exact in fixed-point arithmetic.

The iteration complexity depends on two geometric inputs: a Euclidean ball containing a minimizer and an $\ell_2$ Lipschitz constant. An invariance property of the dual lets us choose a bounded optimal representative, which yields the ball.
\begin{proposition}\label{prop:multi-ot-dual-invariant}
Suppose $y$ is a feasible solution to \eqref{e:MOT_dual_tensor} and $z\in\R{m}$ satisfies $\sum_{k=1}^mz_k=0$. Then $y'$ defined by $y'_k\teq y_k+z_k\mathbf 1_{n_k}$ is also a feasible dual solution with the same objective value.
\end{proposition}
\begin{proof}
Since $\mu_k^\top\mathbf 1_{n_k}=1$,
\[
    b^\top y'=b^\top y+\sum_{k=1}^mz_k\mu_k^\top\mathbf 1_{n_k}=b^\top y+\sum_{k=1}^mz_k=b^\top y.
\]
Moreover, $\sum_ky'_{k,i_k}=\sum_ky_{k,i_k}+\sum_kz_k=\sum_ky_{k,i_k}\le C_{\idx}$ for every $\idx\in I$.
\end{proof}

\begin{proposition}\label{prop:multi-ot-opt-dual-norm-bound}
Assume that $C\ge0$ and that every marginal has full support. There exists an optimal solution $y^\star$ to \eqref{e:MOT_dual_tensor} such that $\|y^\star\|_\infty\le\|C\|_\infty$.
\end{proposition}
\begin{proof}
By Proposition~\ref{prop:multi-ot-dual-invariant}, given an optimal dual solution $\hat y^\star$, we can construct another optimal solution $y^\star$ whose blocks have a common maximum,
\[
    \max_{j\in I_k}y^\star_{k,j}=\max_{\ell\in[m],\,j\in I_\ell}y^\star_{\ell,j}\teq y^\star_{\max}
    \qquad(k\in[m]),
\]
by setting $z_k=-\max_j\hat y^\star_{k,j}+\frac1m\sum_{\ell=1}^m\max_j\hat y^\star_{\ell,j}$ and $y^\star_k=\hat y^\star_k+z_k\mathbf 1_{n_k}$. Let $C_{\max}\teq\max_{\idx}C_{\idx}$, $C_{\min}\teq\min_{\idx}C_{\idx}$, and $i^+_k\in\arg\max_{j\in I_k}y^\star_{k,j}$. The dual constraint at $(i^+_1,\ldots,i^+_m)$ gives
\begin{equation}\label{eq:multi-ot-dual-upper-bound}
    my^\star_{\max}=\sum_{k=1}^my^\star_{k,i^+_k}\le C_{\max}.
\end{equation}
On the other hand, let $X^\star$ be an optimal primal solution, and let $(k,j)$ attain $y^\star_{\min}\teq\min_{\ell,j'}y^\star_{\ell,j'}$. Since $\mu_{k,j}>0$, some $\idx$ with $i_k=j$ has $X^\star_{\idx}>0$, and complementary slackness gives $\sum_\ell y^\star_{\ell,i_\ell}=C_{\idx}$. Hence
\begin{equation}\label{eq:multi-ot-dual-lower-bound}
    y^\star_{\min}
    =C_{\idx}-\sum_{\ell\ne k}y^\star_{\ell,i_\ell}
    \ge C_{\min}-(m-1)y^\star_{\max}
    \ge C_{\min}-\frac{m-1}{m}C_{\max}.
\end{equation}
Since $C\ge0$, \eqref{eq:multi-ot-dual-upper-bound} and \eqref{eq:multi-ot-dual-lower-bound} give $-C_{\max}\le y^\star_{\min}\le y^\star_{\max}\le C_{\max}$, that is, $\|y^\star\|_\infty\le\|C\|_\infty$.
\end{proof}

\begin{lemma}\label{lem:euclidean_value_geometry}
Assume that $C\ge0$ and that every marginal has full support, and set $R_C\teq\|C\|_\infty\sqrt M$. Then $\mathcal B^M_2(R_C)$ contains an optimal solution to \eqref{e:MOT_dual_tensor}, $F(y)\ge-\mathsf{MOT}(C,\musf)$ for every $y\in\R{M}$, and
\begin{equation}\label{e:dual_to_MOT}
    \min_{y\in\mathcal B^M_2(R_C)}F(y)=-\mathsf{MOT}(C,\musf).
\end{equation}
Moreover, if $\idx$ maximizes $\phi_{\idx}(y)$, then $a_{\idx}-b\in\partial F(y)$ and $\|a_{\idx}-b\|_2\le2\sqrt m$.
\end{lemma}
\begin{proof}
By Proposition~\ref{prop:multi-ot-opt-dual-norm-bound}, there is an optimal dual solution $y^\star$ with $\|y^\star\|_\infty\le\|C\|_\infty$, so $\|y^\star\|_2\le R_C$. For every $y$, the vector $y-r(y)\mathbf 1^{(1)}$ is dual feasible with objective $-F(y)$ by \eqref{e:first_block_shift}, so weak duality gives $F(y)\ge-\mathsf{MOT}(C,\musf)$. Conversely, $y^\star$ is feasible, so $r(y^\star)\le0$. If $r(y^\star)<0$, then $y^\star-r(y^\star)\mathbf 1^{(1)}$ would be dual feasible with objective $b^\top y^\star-r(y^\star)>\mathsf{MOT}(C,\musf)$, contradicting strong duality. Hence $r(y^\star)=0$ and $F(y^\star)=-b^\top y^\star=-\mathsf{MOT}(C,\musf)$, proving \eqref{e:dual_to_MOT}.

If $\idx$ maximizes $\phi_{\idx}(y)$, then for every $z\in\R{M}$,
\[
    F(z)\ge\phi_{\idx}(z)-b^\top z
    =\phi_{\idx}(y)-b^\top y+\langle a_{\idx}-b,z-y\rangle
    =F(y)+\langle a_{\idx}-b,z-y\rangle.
\]
Hence $a_{\idx}-b\in\partial F(y)$. Each column of $A$ has $m$ ones, so $\|a_{\idx}\|_2=\sqrt m$, and $\|b\|_2^2=\sum_k\|\mu_k\|_2^2\le m$ because the marginals are probability vectors. Thus $\|a_{\idx}-b\|_2\le2\sqrt m$.
\end{proof}

\subsection{The quantum value algorithm}
\begin{lemma}\label{lem:quantum_max}
Let $y\in\R{M}$ be a classical fixed-point vector. One coherent evaluation of $\idx\mapsto\phi_{\idx}(y)$ uses one query to \eqref{e:cost_oracle_main}, $m$ lookups of $y$, and $O(m)$ exact fixed-point additions, followed by uncomputation. For every $q\in(0,1)$, there is a quantum algorithm that outputs, with probability at least $1-q$, an index $\idx^\star$ with $\phi_{\idx^\star}(y)=\max_{\idx}\phi_{\idx}(y)$, together with this exact value, using $O(\sqrt N\log(1/q))$ such evaluations and exact comparisons.
\end{lemma}
\begin{proof}
The costs and the entries of $y$ are fixed-point numbers, so $\phi_{\idx}(y)$ is computed exactly, and comparisons between two such values are exact bit comparisons. Quantum maximum finding~\cite{durr1996quantum,van2018improvements} over the $N$ tensor indices, boosted to failure probability $q$ by repetition and retention of the best index found, uses $O(\sqrt N\log(1/q))$ evaluations and comparisons and returns an index of maximal value. One more evaluation gives the value.
\end{proof}

\begin{theorem}[Quantum subgradient method for MOT]\label{t:quantum_md}
Let $(C,\musf,\varepsilon)$ be a represented instance as in Section~\ref{s:finite_input_appendix}, with $n_k\le n$ for every $k$, and let $\kappa=\max\{1,\Delta_C/\varepsilon\}$. Under coherent cost access \eqref{e:cost_oracle_main} and coherent lookup of a length-$M$ classical vector, there is a quantum algorithm which, with probability at least $2/3$, outputs $\widehat y\in\R{M}$ and $\widehat V\teq b^\top\widehat y$ such that
\[
    A^\top\widehat y\le c,
    \qquad
    \mathsf{MOT}(C,\musf)-\varepsilon\le\widehat V\le\mathsf{MOT}(C,\musf).
\]
It uses $\widetilde O(m^2n^{m/2+1}\kappa^2)$ coherent cost queries and $\widetilde O(m^3n^{m/2+1}\kappa^2)$ working-vector lookups and unit-cost arithmetic operations on fixed-point words of $O(b_C+b_\mu+b_\varepsilon+\log(mn\kappa))$ bits.
\end{theorem}
\begin{proof}
Quantum minimum and maximum finding~\cite{durr1996quantum}, jointly boosted to failure probability at most $1/12$, compute $L=\min C$ and $U=\max C$ with $\widetilde O(\sqrt N)$ cost queries. If $\Delta_C=U-L\le\varepsilon$, output $\widehat y\teq L\mathbf 1^{(1)}$. By \eqref{e:first_block_shift}, $A^\top\widehat y=L\mathbf 1_N\le c$ and $\widehat V=L$, while $L\le\mathsf{MOT}(C,\musf)\le U\le L+\varepsilon$.

Assume henceforth that $\Delta_C>\varepsilon$, so $\kappa=\Delta_C/\varepsilon$. Work with $C^+\teq C-L\mathbf 1\in[0,\Delta_C]^{I}$: a query to $C^+$ uses one query to $C$ and one exact subtraction, $\|C^+\|_\infty=\Delta_C$, and $\mathsf{MOT}(C^+,\musf)=\mathsf{MOT}(C,\musf)-L$. Define $F$, $\phi$, and $r$ by \eqref{e:value_dual_objective} with $C^+$ in place of $C$.

\emph{Parameters.} Let $G\teq2\sqrt m$ and $R\teq\Delta_C\sqrt M$, and let $T$ be the smallest power of two with $T\ge144\,mM\kappa^2$, so that $3GR/(2\sqrt T)\le\varepsilon/4$. Let $\eta$ be the largest power of two with $\eta\le\eta_\star\teq R/(G\sqrt T)$; it is found exactly by comparing $4^jG^2T$ with $R^2=\Delta_C^2M$ over integers $j$. Set $e\teq\varepsilon/(40G\sqrt T)$. Iterates are stored on the grid $2^{-b_y}\mathbb Z^M$, where $b_y$ is the least integer such that $\eta$ and every $\eta\mu_{k,a}$ lie on the grid and $\sqrt M\,2^{-b_y}\le e/2$.

\emph{Iterations.} Start from $y^{(0)}=0$. At iteration $t<T$, apply Lemma~\ref{lem:quantum_max} with failure probability $q\teq1/(12(T+1))$ to obtain $\idx_t\in\arg\max_{\idx}\phi_{\idx}(y^{(t)})$, and set $g^{(t)}\teq a_{\idx_t}-b$. By Lemma~\ref{lem:euclidean_value_geometry}, $g^{(t)}\in\partial F(y^{(t)})$ and $\|g^{(t)}\|_2\le G$. The point $z\teq y^{(t)}-\eta g^{(t)}$ lies exactly on the grid, and $\|z\|_2\le R+\eta G\le2R$. Let $\Pi$ denote projection onto $\mathcal B^M_2(R)$. If $\|z\|_2^2\le R^2$, an exact rational comparison, set $y^{(t+1)}\teq z=\Pi(z)$. Otherwise compute, by directed-rounding square root and division, a number $\tilde s$ with $R/\|z\|_2-e/(4R)\le\tilde s\le R/\|z\|_2$, and let $y^{(t+1)}$ be $\tilde sz$ with every coordinate rounded toward zero onto the grid. Then $\|y^{(t+1)}\|_2\le\tilde s\|z\|_2\le R$, and since $\Pi(z)=(R/\|z\|_2)z$,
\[
    \|y^{(t+1)}-\Pi(z)\|_2
    \le\sqrt M\,2^{-b_y}+\frac{e}{4R}\|z\|_2\le e.
\]

\emph{Accuracy.} Suppose that all calls to Lemma~\ref{lem:quantum_max} succeed, so that Lemma~\ref{t:qpsm_prelim} applies with this $e$. Since $\eta_\star/2<\eta\le\eta_\star$,
\[
    \frac{R^2}{2\eta T}+\frac{\eta G^2}{2}\le\frac{3GR}{2\sqrt T}\le\frac\varepsilon4.
\]
Moreover, $e\le\varepsilon<\Delta_C\le R$ and $\eta>R/(2G\sqrt T)$, so
\[
    \frac{4Re+e^2}{2\eta}\le\frac{5Re}{2\eta}<5eG\sqrt T=\frac\varepsilon8.
\]
Because $T$ is a power of two, the average $\bar y\teq T^{-1}\sum_{t=0}^{T-1}y^{(t)}$ is exact on a grid with $\log_2T$ additional fractional bits, and Lemma~\ref{lem:euclidean_value_geometry} gives
\[
    -\mathsf{MOT}(C^+,\musf)\le F(\bar y)\le-\mathsf{MOT}(C^+,\musf)+\frac{3\varepsilon}8.
\]

\emph{Certificate.} Apply Lemma~\ref{lem:quantum_max} once more to compute $r(\bar y)$ exactly, and output
\[
    \widehat y\teq\bar y-\bigl(r(\bar y)-L\bigr)\mathbf 1^{(1)}.
\]
By \eqref{e:first_block_shift}, $(A^\top\widehat y)_{\idx}=\phi_{\idx}(\bar y)-r(\bar y)+C_{\idx}\le C_{\idx}$ for every $\idx\in I$, so $A^\top\widehat y\le c$, and
\[
    \widehat V=b^\top\widehat y=-F(\bar y)+L\ge\mathsf{MOT}(C,\musf)-\frac{3\varepsilon}8.
\]
Weak duality gives $\widehat V\le\mathsf{MOT}(C,\musf)$. All arithmetic in the certificate is exact. At most $T+1$ maximum-finding calls are made, so they all succeed except with probability at most $1/12$; together with range finding, the algorithm succeeds with probability at least $5/6$.

\emph{Complexity.} Each call to Lemma~\ref{lem:quantum_max} uses $O(\sqrt N\log T)$ cost queries and $O(m\sqrt N\log T)$ lookups and additions, and each iteration performs $O(M)$ further classical operations. Since $T=O(mM\kappa^2)=O(m^2n\kappa^2)$ and $N\le n^m$, the algorithm uses $\widetilde O(m^2n^{m/2+1}\kappa^2)$ cost queries and $\widetilde O(m^3n^{m/2+1}\kappa^2+TM)$ other operations, where $TM=O(m^3n^2\kappa^2)$ is dominated for $m\ge2$. Finally, $|y_{k,a}|\le R$ and $\log(1/\eta),\log(1/e)=O(b_\varepsilon+\log(mn\kappa))$, which gives the stated word length.
\end{proof}
\end{document}